\documentclass[11pt]{article}
 \usepackage[margin=0.97in]{geometry}
\usepackage{microtype,tikz,amsthm,amssymb,amsmath}
\usepackage{epsfig}
\usepackage{algorithmic,algorithm,complexity}
\newtheorem{defn}{Definition} 
\newtheorem{obs}{Observation}
\newtheorem{theorem}{Theorem}
\newtheorem{lemma}{Lemma}

\newtheorem{corollary}{Corollary}
\newcommand{\rop}{{\sf ROP}}
\newcommand{\rof}{{\sf ROF}}
\newcommand{\Q}{{\mathcal{Q}}}
\newcommand{\rank}{\operatorname{rank}}
\newcommand{\perm}{{\sf perm}}
\newcommand{\var}{\sf var}
\newcommand{\sep}{ rank\text{-}(1,2)\text{-}separator}
\newcommand{\valu}{{\sf value}}

\newcommand{\pcm}{\widehat{M}}

\newcommand{\mrank}{\operatorname{maxrank}}
\usepackage{epsfig}

\title{ Limitations of sum of products of Read-Once Polynomials}
\author{C. Ramya\\Department of Computer Science and Engineering \\
IIT Madras, Chennai INDIA\\
 \texttt{ramya@cse.iitm.ac.in} \and B. V. Raghavendra Rao \\Department of Computer Science and Engineering \\ 
IIT Madras, Chennai INDIA \\
  \texttt{bvrr@cse.iitm.ac.in}}

\begin{document}

\maketitle

\begin{abstract}
 We study limitations of polynomials computed by depth two circuits built over  read-once polynomials (ROPs) and depth three syntactically multi-linear formulas.    
 We prove an exponential lower bound for the size of the $\Sigma\Pi^{[N^{1/30}]}$ arithmetic circuits built over syntactically multi-linear $\Sigma\Pi\Sigma^{[N^{8/15}]}$  arithmetic circuits computing a product of variable disjoint linear forms on $N$ variables.  We extend the result to the case of $\Sigma\Pi^{[N^{1/30}]}$ arithmetic circuits built over ROPs of unbounded depth, where  the number of variables with $+$ gates as a parent in an proper sub formula is bounded by $N^{1/2+1/30}$. We show that the same lower bound holds for the permanent polynomial. Finally we obtain an exponential lower bound for the sum of ROPs computing a  polynomial in ${\sf VP}$ defined by Raz and Yehudayoff~\cite{RY09}.

 Our results demonstrate a class of formulas of unbounded depth with exponential size lower bound against the permanent and can be seen as an exponential improvement over the  multilinear formula size lower bounds given by Raz~\cite{Raz04a} for a sub-class of multi-linear and non-multi-linear formulas.  
 Our proof techniques are built on the one developed by Raz~\cite{Raz04a} and later extended by Kumar et. al.~\cite{KMS13} and are based on  non-trivial analysis of ROPs under random partitions.    Further, our results exhibit  strengths and limitations of the lower bound techniques introduced by Raz~\cite{Raz04a}.
\end{abstract}

\newpage
\section{Introduction}
\label{sec:intro}
 More than three decades ago,   Valiant~\cite{Val79} developed the theory of Algebraic Complexity classes based on arithmetic circuits as the model of algebraic computation.    Valiant considered the   permanent polynomial ${\perm}_n$ defined over an $n\times n $ matrix $X= (x_{i,j})_{1\le i,j\le n}$ of variables:
\begin{center}
%\begin{align}
$\perm_n(X) = \sum\limits_{\pi\in {S}_n}\prod\limits_{i=1}^n x_{i,\pi(i)},$
\label{perm}
%\end{align}
\end{center}

where $S_n$ is the set of all permutations on $n$ symbols.  Valiant~\cite{Val79} showed that the polynomial family $(\perm_n)_{n\ge 0}$ is complete for the complexity class ${\sf VNP}$. Further, Valiant~\cite{Val79} conjectured that ${\perm}_n$ does not have polynomial size arithmetic circuits.  Since then, obtaining super polynomial size lower bounds for arithmetic circuits computing  $\perm_n$ has been a pivotal problem in Algebraic Complexity Theory. However, for general classes of arithmetic circuits,  the best known lower bound is quadratic in the number of variables~\cite{MR04}. 

Naturally, the focus has been on proving lower bound for $\perm_n$ against restricted classes of circuits.  Grigoriev and Karpinski~\cite{GK98} proved an exponential size lower bound for depth three circuits of constant size over finite fields.  Agrawal and Vinay~\cite{AV08} (See also~\cite{Tav13,Koi12}) showed that proving exponential lower bounds against depth four arithmetic circuits  is enough to resolve Valiant's conjecture, and hence explaining the lack of progress in extending the results in~\cite{GK98} to higher depth circuits.  This was strengthened to depth three circuits over infinite fields by Gupta et. al.~\cite{KKS13}. 

Recently,  Gupta et. al.~\cite{GKKS14} obtained  a   $2^{\Omega(\sqrt{n}\log n)}$ size lower  bound for homogeneous depth four circuits computing $\perm_n$  where  the bottom fan-in is bounded by $O(\sqrt{n})$. The techniques introduced in~\cite{GKKS14,Kay12} have been generalized and applied to prove lower bounds against various classes of constant depth arithmetic circuits, regular arithmetic formulas and homogeneous arithmetic formulas.~(See e.g., \cite{KLSS14a,KS14,KS15}.) Exhibiting polynomials that have exponential lower bound against concrete classes of arithmetic circuits is an important research direction.  

 In 2004, Raz~\cite{Raz04a} showed that any multilinear formula computing $\perm_n$ requires size $n^{\Omega(\log n)}$, which was one of the first super polynomial lower bounds against formulas of unbounded depth. Further, in~\cite{Raz04b}, Raz extended this to separate multilinear formulas from multilinear circuits.  Raz's work lead to  several lower bound results, most significant being an exponential separation of constant depth multilinear circuits~\cite{RY09}. More recently, Kumar et. al~\cite{KMS13} extended the techniques developed in~\cite{Raz04a} to prove lower bounds against non-multilinear circuits and formulas.
 \paragraph*{•}
\noindent{\bf Motivation and our Model :}
 Depth three $\Sigma\Pi\Sigma$ circuits are in fact   $\Sigma\Pi$ circuits built over linear forms. A linear form can be seen as the simplest form of read-once formulas ($\rof$):  formulas  where a variable appears at most once as a leaf label. Polynomials computed by $\rof$s are called read-once polynomials or $\rop$s.  There are two natural generalizations of the $\Sigma\Pi\Sigma$ model: 1) Replace linear forms by sparse polynomials, this leads to the well studied $\Sigma\Pi\Sigma\Pi$ circuits; and 2) Replace linear forms with more general read-once formulae, this leads to the class of $\Sigma\Pi$ circuits over read-once formulas or $\Sigma\Pi\rop$ for short.
 
 In this paper, we consider the second extension, i.e., $\Sigma\Pi\rop$.  Restricted forms of 
 $\Sigma\Pi\rop$ were already considered in the literature. For example, Shpilka and Volkovich~\cite{SV08} obtained identity testing algorithms for the sum of $\rop$s. Further \cite{MRS14} gives identity tests  for $\Sigma\Pi\rop$ when the top fan-in is restricted to two.
 
 Apart from being a natural generalization of $\Sigma\Pi\Sigma$ circuits, the class  $\Sigma\Pi\rop$ can be  seen as  building non-multi-linear polynomials using the simplest possible multi-linear polynomials viz. $\rop$s. 
  \paragraph*{•}
\noindent{\bf Our Results : }
We study the limitations of the model  $\Sigma\Pi\rop$ for some restricted class of circuits. Firstly, we prove,
\begin{theorem}
\label{thm:depth-three}
Let $f_{i,j}$ be  $N$-variate $\Sigma\Pi\Sigma$ syntactic multi-linear formulas with bottom $\Sigma$-fan-in at most $ N^{1/2+ \lambda}$ where $\lambda\le 1/30$, and top $\Sigma$-fan-in  at most $s'$ for   $1\le i
\le s$ and $1\le j\le t$.  Also assume that $t\le N^{1/30}$. There is a  product of variable disjoint linear forms  $p_{lin}$ such that, if 
$ \sum_i \prod_j f_{i,j} = p_{lin}$
then $s\cdot s'=  2^{\Omega(N^{1/4})}$.
\end{theorem}
Our arguments do not directly generalize to the case of unbounded depth $\rop$s with small bottom $\Sigma$ fan-in.  Nevertheless, we obtain a generalization of Theorem~\ref{thm:depth-three}, allowing $\rop$s of unbounded depth with a more stringent restriction than bottom $\Sigma$-fan-in.  Let $F$ be an $\rof$ and for a gate $v$ in $F$, let ${\sf sum\mbox{-}fan\mbox{-}in}(v)$ be the number of variables in the sub-formula rooted at $v$ whose parents are labelled as $+$. Then $s_F$ is the maximum value of ${\sf sum\mbox{-}fan\mbox{-}in}(v)$, where the maximum is taken over all $+$ gates in $F$ excluding the top layer of $+$ gates. Note that, in the case of $\Sigma\Pi\Sigma$ $\rop$s, $s_F$ is equivalent to the bottom fan-in. For an $\rop$ $f$, $s_f$ is the smallest value of $s_F$ among all $\rof$s  $F$ computing $f$. We prove,    
\begin{theorem}
\label{thm:sum-product-rop} 
Let $f_{i,j}$ be $\rop$s with $s_{f_{i,j}}\le N^{1/2+ \lambda}$  for $\lambda\le 1/30$,  $1\le i
\le s$ and $1\le j\le t$, where $t\le N^{1/30}$. There is a product of linear forms ${ p_{lin}}$  such that, if $ \sum_i \prod_j f_{i,j} = {p_{lin}}$
then $s =  2^{\Omega(N^{1/4})}$.
\end{theorem}

As far as we know, this is the first exponential lower bound for a sub-class of non-multi-linear formulas of unbounded depth.
 It can be noted that our result above does not depend on the depth of the $\rop$s. 
Further, note that even though a product of linear forms is a simple linear  projection of $\perm_n$, Theorem~\ref{thm:sum-product-rop} does not imply a lower bound for $\perm_n$ due to restrictions on $s_F$, since linear  projections might change the bottom fan-in of the resulting $\rop$s. However, we prove,
\begin{theorem}
\label{thm:sum-product-rop-perm} 
Let $f_{i,j}$ be $\rop$s with $s_{f_{i,j}}\le N^{1/2+ \lambda}$ for $\lambda\le 1/30$, $1\le i
\le s$ and $1\le j\le t$, for $N=n^2$ and $t\le N^{1/30}$. Then,  if 
$ \sum_i \prod_j f_{i,j} = \perm_n$
then $s =  2^{\Omega(N^{\epsilon})}$ for some $\epsilon>0$.
\end{theorem}

Finally, we show that the polynomial $g$ defined by Raz-Yehudayoff~\cite{RY08} cannot be   written as sum of sub-exponentially many  $\rop$s:

\begin{theorem}
\label{thm:sum-rop}
There is a polynomial $g\in {\sf VP}$ such that for any $\rop$s $f_1,\ldots, f_s$, if
$\sum_{i}f_i = g $, then we have $s=2^{\Omega(n/\log n)}$.
\end{theorem}

\noindent{\bf Related Results :}
Shpilka and Volkovich~\cite{SV08} proved a linear lower bound for a special class of $\rop$s to sum-represent the polynomial $x_1\cdots x_n$ and used it crucially in their identity testing algorithm. Theorem~\ref{thm:sum-rop} is an exponential lower bound against the same model as in~\cite{SV08}, however against a polynomial in ${\sf VP}$. 
It should be noted that the results in  Raz~\cite{Raz04a} combined with~\cite{KMS13} immediately implies a lower bound of $n^{\Omega(\log n)}$ for the  sum of $\rop$s. Our results are  an exponential improvement of bound given by~\cite{Raz04a}. 

 Kayal~\cite{Kay12} showed that  at least $2^{n/d}$ many polynomials of degree $d$ are required to represent the polynomial $x_1\dots x_n$  as sum of powers. Our model is significantly different from the one in ~\cite{Kay12} since it includes high degree monomials, though the powers are restricted to be sub-linear, whereas Kayal's argument works against arbitrary powers.
\paragraph*{•} 
\noindent{\bf Our Techniques} 
Our techniques are broadly based on the  partial derivative matrix technique introduced by Raz~\cite{Raz04a} and later extended by Kumar et. al~\cite{KMS13}.  It can be noted that the lower bounds obtained in \cite{Raz04a} are super polynomial and not exponential.   Though Raz-Yehudayoff~\cite{RY09} proved exponential lower bounds, their argument works only against bounded depth multilinear  circuits. Further, the arguments in \cite{Raz04a,RY09} do not work for the case of non-multilinear circuits, and fail even in the case of products of two multilinear formulas. This is because  rank of the partial derivative matrix, a complexity measure  used by \cite{Raz04a,RY09} (see Section~\ref{sec:prelim} for a definition) is defined only for multi-linear polynomials.  Even though  this issue can be overcome by a generalization introduced by Kumar et. al~\cite{KMS13}, the limitation  lies in the fact that the upper bound of $2^{n-n^{\epsilon}}$ for an $n^2$ or $2n$ variate polynomial, obtained in \cite{Raz04a} or \cite{RY09} on the measure for the underlying arithmetic formula model is insufficient to handle products of two $\rop$s. 

Our approach to prove Theorems~\ref{thm:sum-product-rop} and \ref{thm:sum-product-rop-perm}  lie in  obtaining an exponentially stronger upper bounds (see Lemma~\ref{lem:rank-sum-product} )  on the rank of the partial derivative matrix of  an $\rop$ $F$  on $N$ variables where   $s_F\le N^{1/2+1/30}$. Our proof  is a technically involved  analysis of the structure of $\rop$s under random partitions of the variables. Even though the restriction on $s_F$ might  look un-natural, in Lemma~\ref{lem:prod-linear}, we show that a simple product of variable disjoint linear forms in $N$-variables, with $s_F\ge N^{2/3}$ achieve exponential rank with probability $1-2^{-\Omega(N^{1/3})}$. Thus our results highlight the strength and limitations of the techniques developed in  \cite{RY09,KMS13}  to the case of  non-multi-linear formulas.   

Finally proof of Theorem~\ref{thm:sum-rop} is based on an observation  pointed out to the authors by an anonymous reviewer. We have included it here since the details have been worked out completely by the authors. 

Due to space limitations, all the missing proofs can be found in Sections~5,6 and 7.

\section{Preliminaries}
\label{sec:prelim}
Let $\mathbb{F}$ be an arbitrary field and $X = \{x_1,\ldots,x_N \}$ be a set of variables. An {\em arithmetic circuit} $\mathcal{C}$ over $\mathbb{F}$ is a directed acyclic graph with vertices of in-degree 0 or 2 and exactly  one vertex of out-degree 0 called the output gate. The vertices of in-degree 0 are called {\em input gates}  and are labeled  by elements from $X \cup \mathbb{F}$. The vertices of in-degree more than 0 are labeled by either $+$ or $\times$. Thus every gate of the circuit naturally computes a polynomial. The polynomial $f$ computed by $\mathcal{F}$ is the polynomial computed by the output gate of the circuit. 

An {\em arithmetic formula} is a an arithmetic circuit $\mathcal{F}$ where every gate has out-degree bounded by $1$, i.e., the underlying undirected graph $\mathcal{F}$ is a tree.

The {\em size} of an arithmetic circuit $\mathcal{F}$   is  the number of gates in $\mathcal{F}$. For any gate $v$  {\em depth} of $v$ is the  length of the longest path from an input gate  to $v$ gate in $\mathcal{F}$. Depth of $\mathcal{F}$ is defined as the depth of its output gate.

An {\em arithmetic read-once formula} ($\rof$ for short)  is an  arithmetic formula $\mathcal{F}$ over $X$ where every input variable $x\in X$ occurs as a label of at most once  $\mathcal{F}$. The polynomial $f$ computed by an $\rof$  $\mathcal{F}$ is called a {\em read-once-polynomial} or $\rop$. 
  
Let $f(y_1,\ldots,y_m,z_1,\ldots,z_m) \in \mathbb{F}[y_1,\ldots, y_m, z_1,\ldots, z_m]$ be a multilinear polynomial. The {\em partial derivative matrix of $f$} denoted by $M_f$~\cite{Raz04a} is a $2^{m}\times 2^{m}$ matrix  defined as follows: The rows of $M_f$ are labeled by all possible  multilinear monomials in $\{y_1,\ldots,y_m\}$ and the columns of $M_f$ be labeled by all possible  multilinear monomials in $\{z_1,\ldots,z_m\}$. For any two  multilinear monomials  $p$ and $q$, the entry $M_f[p,q]$ is the coefficient of  $p\cdot q$ in $f$. 

\begin{lemma}{\em \cite{Raz09}}(Sub-Additivity.)
\label{sub-additivity}
%Let $\Phi$ be an arithmetic formula computing a polynomial over the set of variables $\{y_1,\ldots,y_m\}\cup\{z_1,\ldots,z_m\}$. 
Let $f=f_1+f_2$  where $f, f_1$ and $f_2$ are multilinear polynomials in $\mathbb{F}[y_1,\ldots, y_m, z_1,\ldots, z_m]$.  Then,
$rank(M_f)\leq rank(M_{f_1}) + rank(M_{f_2}).$ Moreover, if  $\var(f_1)\cap \var(f_2)=\emptyset$ then $rank(M_f)= rank(M_{f_1}) + rank(M_{f_2}).$
\end{lemma}

\begin{lemma}{\em \cite{Raz09}}(Sub-Multiplicativity.)
\label{sub-multiplicativity}
%Let $\Phi$ be an arithmetic formula computing a polynomial over the set of variables $\{y_1,\ldots,y_m\}\cup\{z_1,\ldots,z_m\}$. 
Let $f= f_1\times f_2$, where $f,f_1$ and $f_2$ are multilinear polynomials in $\mathbb{F}[y_1,\ldots,y_m, z_1,\ldots,z_m]$, and $\var(f_1)\cap \var(f_2)=\emptyset$. Then, 
$rank(M_f)= rank(M_{f_1}) \cdot rank(M_{f_2}).$ 
\end{lemma}
%\begin{proof}
%Since $\Phi$ is syntactically multilinear and v is a product gate, $v_1$ and $v_2$ are variable disjoint. In this case observe that the matrix $M_v$ is the tensor product of $M_{v_1}$ and $M_{v_2}$. By linear algebra,
%\begin{center}
%$rank(M_v)= rank(M_{v_1}) \cdot rank(M_{v_2})$
%\end{center}
%\end{proof}

Kumar et. al.~\cite{KMS13} generalized the notion of partial derivative matrix to include  polynomials that are not multilinear.
Let $Y = \{ y_1,\ldots,y_m\}$ and $Z=\{z_1,\ldots,z_m\}$.  Let $f\in\mathbb{F}[Y,Z]$ be a polynomial. The {\em polynomial coefficient matrix} of $f$ denoted by $\pcm_f$ is a $2^m\times 2^m$ matrix defined as follows. For  multilinear monomials $p$ and $q$ in variables $Y$ and $Z$ respectively, the entry $\pcm_f [p, q] = A$ if and only if $f$ can be uniquely expressed as $f = pq\cdot A + B$ where $A,B \in \mathbb{F}[Y,Z]$ such that $\var(A)\subseteq \var(p)\cup\var(q)$ and  $B$ does not have any monomial that is divisible by $p\cdot q$ and contains only variables present in $p$ and $q$.

\begin{obs}\cite{KMS13}
For a multilinear polynomial $f\in\mathbb{F}[Y,Z]$, we have $\pcm_f= M_f$.
\end{obs}

Observe that the matrix $\pcm_f$ has polynomial entries. Therefore $\rank(\pcm_f)$ is defined only under a substitution function that substitutes every variable in $f$ to a field element.

For any substitution function $S : Y\cup Z \rightarrow \mathbb{F}$, let us denote by ${\pcm_f}|_S$ the matrix obtained by substituting every variable in $f$ at each entry of $\pcm_f$ to the field element given by $S$.
Define,  
$\mrank(\pcm_f) \triangleq \max\limits_{S : Y\cup Z \rightarrow \mathbb{F}} \rank({\pcm_f}|_S).$
Having defined polynomial coefficient matrix $\pcm_f$ and $\mrank(\pcm_f)$ we now look at properties of $\pcm_f$ with respect to $\mrank$.
\begin{lemma}{\em \cite{KMS13}}(Sub-additivity.)
\label{sub-aditivity2}
Let $f,g \in \mathbb{F}[Y,Z]$. Then, 
$\mrank(\pcm_{f+g}) \leq \mrank(\pcm_f) + \mrank(\pcm_g).$
\end{lemma}
\begin{lemma}{\em \cite{KMS13}}(Sub-multiplicativity.) 
\label{sub-multiplicativity2}
Let $Y_1,Y_2\subseteq Y$ and $Z_1,Z_2\subseteq Z$ such that $Y_1 \cap Y_2 =\emptyset$ and $Z_1 \cap Z_2 =\emptyset$. Then for any polynomials $f\in\mathbb{F}[Y_1,Z_1]$, $g\in \mathbb{F}[Y_2,Z_2]$ we have: $\mrank(\pcm_{fg}) = \mrank(\pcm_f) \cdot \mrank(\pcm_g)$.
\end{lemma}

%%Throughout the paper $X =\{x_{i,j}\}_{\{i,j\}\in[n]}$ denotes a matrix of $n^2$ distinct variables.  
%%Let $Y = \{y_1,\ldots,y_{n^2}\}$ and $Z = \{z_1,\ldots,z_{n^2}\}$ be two distinct sets of variables. 

A partition of $X$ is a function  $\varphi : X \rightarrow Y \cup Z \cup \{0,1\}$ such that $\varphi$ is an injection  when restricted to $Y\cup Z$, i.e., $\forall x\neq x'\in X$, if  $\varphi(x)\in Y\cup Z$ and $\varphi(x')\in Y\cup Z$ then $\varphi(x)\neq \varphi(x')$.  
Let $F$ be an $\rof$ and $\varphi : X \rightarrow Y \cup Z \cup \{0,1\}$ be a partition function. Define $F^\varphi$ to be the formula obtained by replacing every variable $x$ that appears as a leaf in $F$ by $\varphi(x)$. Then the polynomial $f^\varphi$ computed by $F^\varphi$ is $f^\varphi = f(\varphi(x_1),\ldots,\varphi(x_n))$. Observe that $f^\varphi\in\mathbb{F}[Y,Z]$.   

An arithmetic formula $\mathcal{F}$ is said to be a {\em constant-minimal formula} if no gate $u$ in $F$ has both its children to be constants. Observe that for any arithmetic formula $F$, if there exists a gate $u$ in $F$ such that $u= a~{\sf op}~b, a,b\in\mathbb{Z}$ then  we can replace $u$ in $F$ by the constant $a~ {\sf op}~ b$, where ${\sf op}\in\{+,-\}$. Thus we  assume without loss of generality that $\mathcal{F}$ is constant-minimal.  

We need some observations on formulas that compute natural numbers.  Recall that an arithmetic formula $\mathcal{F}$ is said to be monotone if $\mathcal{F}$ does not contain any negative constants.

Let $G$ be a monotone arithmetic formula were the leaves are labeled numbers in $\mathbb{N}$. Then for any gate $v$ in $G$, the value of $v$ denoted by $\valu(v)$ is  defined  as : If $u$ is a leaf then $\valu(u) = a$ where $a\in\mathbb{N}$ is the label of $u$ and  $u = u_1 ~{\sf op}~ u_2$  then $\valu(u)=\valu(u_1)~{\sf op}~\valu(u_2)$, where ${\sf op}\in \{+,\times\}$.
% If $u$ is a $\times$ gate then $\valu(u)=\valu(u_1)\times\valu(u_2)$. 
 Finally, $\valu(G)$ is the value of the  output gate of $G$. 
Let $G$ be a monotone arithmetic formula with leaves labelled by either $1$ or $2$. A node $u$ a is called a $\sep$ if $u$ is a leaf and $\valu(u)=2$ or $u = u_1 + u_2 $ with $\valu(u)=2$ and $\valu(u_1)=\valu(u_2)=1$.
The following is a simple upper bound on the value computed by a formula.

\begin{lemma}
\label{max-value}
Let $G$ be a binary monotone arithmetic formula with $t$ leaves. If every leaf in $G$ takes a value at most $N>1$, then $\valu(G)\leq N^t$.
\end{lemma}
\begin{proof}
The proof is by induction on the size of the formula. Base Case : $s=1$ \begin{itemize}
\item If $G$ has a single $+$ gate then $\value(G)\leq N+N \leq N^2$.
\item If $G$ has a single $\times$ gate then $\value(G)\leq N\cdot N = N^2$.
\end{itemize}
Induction Step : Let $u$ be the output gate of $G$ with children $u_1$ and $u_2$. Let the number of leaves in the sub formula rooted at $u_1$ and $u_2$ be $t_1$ and $t_2$.
\begin{itemize}
\item If $u$ is a $+$ gate. Then, $\valu(u)=  \valu(u_1)+\valu(u_2)$. By induction hypothesis, $\valu(u)\leq N^{t_1}+ N^{t_2} \leq N^{t_1+t_2} \leq N^t$.
\item If $u$ is a $\times$ gate. Then, $\valu(u)= \valu(u_1)\times \valu(u_2)$. By induction hypothesis, $\valu(u)\leq N^{t_1}\times N^{t_2} \leq N^{t_1+t_2} \leq N^t$.
\end{itemize}\end{proof}
Any formula with a large value should have a large number of $\sep$s.

\begin{lemma}
\label{lem:rank-2nodes}
Let $F$ be a binary monotone arithmetic formula with leaves labeled by either $1$ or $2$. If $\valu(F)>2^r$ then there exists at least $\frac{r}{\log N}$ nodes that are $\sep$s.
\end{lemma}
\begin{proof}
Let $F$ be a binary monotone arithmetic formula with leaves labeled by either $1$ or $2$.  First mark every node $u$ such that $u$ is a $\sep$ and remove sub-formula rooted at $u$ except $u$. Consider any leaf $v$
that remains unmarked and along the path from $v$ to root there is no node that is marked. Then $\valu(v)=1$.  Consider the unique path from $v$ to root in $F$. Let $p$ the first node in the path such that $\valu(p)\geq 2$. Let $p_1$ and $p_2$ be the children of $p$. Without
loss of generality let $p_1$ be an ancestor of $v$. Then observe that there is atleast one marked node(say $q$) in the sub-formula rooted at $p_2$. Set $\valu(q)=\valu(q)+1$. If $p$ is a $+$ gate set $p_1=0$ else set $p_1=1$.  Let $u_1,\ldots,u_t$ be the leaves of the resulting
formula at the end of this process. For every $1 \leq i\leq t$, we have $2\le \valu(u_i)\leq N$. Therefore by Lemma~\ref{max-value}, $\valu(F)\leq N^t$. Since $\valu(F)>2^r$, we have $2^r < N^t$. Therefore $t>\frac{r}{\log N}$ as required.
\end{proof}

Finally, we will use the following variants of Chernoff-Hoeffding bounds.
\begin{theorem}{\em \cite{DP09}}(Chernoff-Hoeffding bound)
\label{chernoff}
Let $X_1,X_2,\ldots,X_n$ be independent random variables. Let $X=X_1+X_2+\cdots+X_n$ and $\mu = \mathbb{E}[X]$. Then for any $\delta > 0$,
\begin{itemize}
\item[(1)] $\Pr[X>(1+\delta)\mu]<\left(\frac{\mathrm{e}^\delta}{(1+\delta)^{(1+\delta)}}\right)^{\mu}$; and
\item[(2)] $\Pr[X \geq (1+\delta)\mu] \leq \mathrm{e}^\frac{-\delta^2\mu}{3}$; and
\item[(3)] $\Pr[X \leq (1-\delta)\mu] \leq \mathrm{e}^\frac{-\delta^2\mu}{2}$.
\end{itemize}
\end{theorem}
%\paragraph*{Notations and constants:} 

\section{Hardness of representation for Sum of $\rop$s}
%\documentclass[a4paper,UKenglish]{lipics}
%%This is a template for producing LIPIcs articles. 
%%See lipics-manual.pdf for further information.
%%for A4 paper format use option "a4paper", for US-letter use option "letterpaper"
%%for british hyphenation rules use option "UKenglish", for american hyphenation rules use option "USenglish"
%% for section-numbered lemmas etc., use "numberwithinsect"
% 
%\usepackage{microtype,tikz}
%\usepackage{epsfig}
%\usepackage{complexity}
%\usepackage{algorithmic,algorithm}
%\newtheorem{defn}{Definition} 
%\newtheorem{obs}{Observation}
%\newtheorem{claim}{Claim}
%\newcommand{\rop}{{\sf ROP}}
%\newcommand{\rof}{{\sf ROF}}
%\newcommand{\rank}{\operatorname{rank}}
%\newcommand{\perm}{{\sf perm}}
%\newcommand{\var}{\sf var}
%\newcommand{\sep}{(1,2)- separator}
%\newcommand{\valu}{{\sf value}}
%\newcommand{\mon}{\sf mon}
%\newcommand{\cfrop}{{\sf CF\mbox{-}ROP}}
%\newcommand{\cfrof}{{\sf CF\mbox{-}ROF}}
%\newcommand{\pcm}{\widehat{M}}
%\newcommand{\multdeg}{{\sf multdeg}}
%\newcommand{\mrank}{\operatorname{maxrank}}
%\usepackage{epsfig}
%
%
%\bibliographystyle{plain}% the recommended bibstyle
%
%% Author macros::begin %%%%%%%%%%%%%%%%%%%%%%%%%%%%%%%%%%%%%%%%%%%%%%%%
%\title{Sum of products of $\rop$s}
%\titlerunning{Sum of products of $\rop$s} 
%
%\begin{document}
%
%\section*{Roadmap :}

%
%Let $f$ be and $\rop$ computed by an $\rof$ $F$.
%
%Suppose $f_1 + f_2 + \cdots + f_s = g$ then $s=2^{\Omega(n)}$.

Let $X=\{x_1,\ldots,x_{2n}\}, Y=\{y_1,\ldots,y_{2n}\}, Z=\{z_1,\ldots,z_{2n}\}$. Define  $\cal{D'}$ as a distribution on the functions $\varphi:X\to Y\cup Z$ as follows : For $1\leq i\leq 2n$,
\begin{center}
$\varphi(x_{i}) \in  \begin{cases} Y &\mbox{with prob. } \frac{1}{2} \\
Z & \mbox{with prob. } \frac{1}{2} \end{cases}$
\end{center}
 Observe that $|\varphi(X)\cap Y| = |\varphi(X)\cap Z|$ is not necessarily true. % However for any variable $x_{i}$, $\varphi(x_{i})\in\{Y,Z\}$ independent of the other variables.
%\begin{defn}
Let $F$ be a binary arithmetic formula computing a polynomial $f$ on the variables $X=\{x_1,\ldots,x_{2n}\}$. Note that any gate with at least one variable as a child can be classified as:
\begin{itemize}
\item[(1)]type-$A$ gates :  sum gates both of whose children are variables,
\item[(2)]type-$B$ gates :  product gates both of whose children are variables,
\item[(3)]type-$C$ gates :  sum gates exactly one child of which is a variable and the other an internal gate; and
\item[(4)]type- $D$ gates: product gates exactly one child of which is a variable and the other an internal gate
\end{itemize}
 
%\end{defn}

%
%Observe that we have not considered a product gate $g$ such that exactly one child of $g$ is a leaf (say a variable $x$) and the other child is a sub-formula (say $h$) as $\rank(M_{g^\varphi}) = \rank(M_{{(x\cdot h)}^\varphi}) \leq \rank(M_{h^\varphi})$.

Given any $\rof$ $F$, let there be $a$ type-$A$ gates, $b$ type-$B$, $c$ type-$C$  and $d$ type-$D$ gates in $F$. Note that $2a+2b+c+d = 2n$.  Let $\varphi\sim \mathcal{D}'$. Let there be $a'$ gates of type-$A$ that achieve $\rank$-1 under $\varphi$ and let $a''$ gates of type-$A$ that achieve $\rank$-2 under $\varphi$. Then, $a=a'+a''$.

\begin{lemma}\footnote{A brief outline of the proof of Lemma~\ref{lem:rank-upper bound} was suggested by an anonymous reviewer, the details included here for completeness and since the details were worked out completely by the authors.}
\label{lem:rank-upper bound}
Let $F$ be an $\rof$ computing an $\rop$ $f$ and $\varphi : X\to Y \cup Z$. Then,
$\rank(M_{f^\varphi})\leq  2^{a''+\frac{a'}{2}+{b} + \frac{c}{2}}. $
\end{lemma}
\begin{proof}
Observe that for any type-$D$ gate $g=h\times x$,  $\rank(M_{g^\varphi}) = \rank(M_{{(x\cdot h)}^\varphi}) =\rank(M_{h^\varphi})$, and hence type-$D$ gates do not contribute to the rank.

The proof is by induction on the structure of $F$. 
Base case is when $F$ is of depth 1. 
Let $r$ be the root gate of $F$ computing the polynomial $f$. Then 
\begin{itemize}
\item $r$ is an type-$A$ gate with children $x_1,x_2$ :
 $f = x_1 + x_2$. For any $\varphi$, we have $\rank(M_{f^\varphi})\leq 2 $. Then $a=1,b=0,c=0$. Therefore either $a'=1$ or $a''=1$. In either case, $\rank(M_{f^\varphi})\leq 2^{a''+\frac{a'}{2}+{b}} $.  
\item $r$ is a type-$B$ gate with children $x_1,x_2$ :
 $f = x_1 \cdot x_2$. For any $\varphi$ we have $\rank(M_{f^\varphi})\leq 1 $. Then $a=0,b=1,c=0$. Therefore $\rank(M_{f^\varphi})\leq  2^{a''+\frac{a'}{2}+{b} +\frac{c}{2}} $. 
\end{itemize}
For the induction step, we have the following cases based on the structure of $f$.
\begin{itemize}
\item $r$ is a type-$C$ gate with children $x, h$, i.e.,  $f = h + x$. For any $\varphi$, we have by sub-additivity 
 $\rank(M_{f^\varphi})\leq  \rank(M_{h^\varphi}) + \rank(M_{x^\varphi})$. Let $a_h',a_h''$ be the number of type-$A$ gates in the sub-formula rooted at $h$ that achieve rank-1 and rank-2 under $\varphi$ respectively. Let $b_h,c_h$ be the number of type-$B$ and $c$ type-$C$ gates in the sub-formula rooted at $h$. We now have $a'=a_h',a''=a_h'',b=b_h, c=c_h+1$, and
$\rank(M_{f^\varphi}) \leq  \rank(M_{h^\varphi}) + \rank(M_{x^\varphi}).
$
By Induction hypothesis $ \rank(M_{h^\varphi})\leq  2^{a_h''+\frac{a_h'}{2}+b_h+\frac{c_h}{2}}$. First suppose the case when  $a_h''+\frac{a_h'}{2}+b_h+\frac{c_h}{2} \ge 1.5$, then, $ 
\rank(M_{f^\varphi}) \leq  2^{a_h''+\frac{a_h'}{2}+b_h+\frac{c_h}{2}} + \rank(M_{x^\varphi}) = 2^{a_h''+\frac{a_h'}{2}+b_h+\frac{c_h}{2}} + 1
\leq 2^{a''+\frac{a'}{2}+{b}+ \frac{c}{2}}$.
Now suppose $a_h''+\frac{a_h'}{2}+b_h+\frac{c_h}{2} <1.5$ and hence $ a_h''+\frac{a_h'}{2}+b_h+\frac{c_h}{2} \le 1$ (since $a_h'', a_h, b_h$ and $c_h$ are integers).  If $a_h'' =1$, then $rank(M_{f^\varphi}) =2 <  2^{a''+\frac{a'}{2}+\frac{b}{2}+ \frac{c}{2}}$. Finally, if $a_h'' =0$,  for all of the remaining  possibilities, we have $\rank(M_{f^{\varphi}}) \le 2 \le 2^{a''+\frac{a'}{2}+b+ \frac{c}{2}}$.
%Therefore $\rank(M_{f^\varphi})\leq 2^{a''+\frac{a'}{2}+\frac{b}{2}} $. 
\item $r= g\times h $ be an internal gate. For $H\in \{g,h \}$, let  $a_H',a_H''$ be the number of type-$A$ gates that achieve rank-1 and rank-2 under $\varphi$ respectively and  $b_H,c_H$ be the number of type-$B$ and $c$ type-$C$ gates in the sub-formula rooted at $H$. 
Then $f=g*h$ where $*\in\{+,\times\}$. In either case, 
$\rank(M_{f^\varphi}) \leq \rank(M_{g^\varphi})\cdot \rank(M_{h^\varphi})$, and from Induction hypothesis $
\rank(M_{f^\varphi}) \leq \cdot 2^{a_g''+\frac{a_g'}{2}+{b_g} + \frac{c_g}{2}} 2^{a_h''+\frac{a_h'}{2}+b_h+ \frac{c_h}{2}}.$
Since $a'=a_g'+a_h',a''=a_g''+a_h'',b=b_h+b_g,c=c_g+c_h$ we have 
$\rank(M_{f^\varphi})\leq  2^{a''+\frac{a'}{2}+b+ \frac{c}{2}}.$
\end{itemize}
\vspace{-4mm}
\end{proof}

\begin{lemma}
\label{lem:a-type}
Let $F$ be a $\rof$ and $\varphi\sim\mathcal{D}'$. Let $a'$ be the number type-$A$ gates that achieve $\rank$-$1$ under $\varphi$. Then,
$\Pr_{\varphi\sim\mathcal{D}'}\left[\frac{2}{5}a \le a' \le \frac{3}{5} a \right] \ge 1- 2^{-\Omega(a)}.$
\end{lemma}
\begin{proof}
Let $v$ be a type-$A$ gate in $F$. Then $f_v=x_i+x_j$ for some $i,j\in[N]$. Then $\Pr[ \rank(M_{f_{v}^\varphi})=1] = \Pr[(\varphi(x_i), \varphi(x_j)\in Z ) \vee (\varphi(x_i), \varphi(x_j)\in Y)] = \frac{1}{2}$. Therefore,
$\mathbb{E}[a'] = a/2$. Applying Theorem~\ref{chernoff} (2) and (3)  with $\delta=1/2$, we get the required bounds for $a'$.
\end{proof}
\begin{lemma}
\label{lem:rank-general}
Let $F$ be a $\rof$ computing and $\rop$ $f$ $2n$ variables,  and $\varphi\sim {\cal D}'$. Then with probability at least $1- 2^{-\Omega(n)}$,
$\rank(M_{f^{\varphi}})\le 2^{n-\frac{n}{5\log n}}$.
\end{lemma}
\begin{proof} 
Consider the following two cases: \\
\noindent\textbf{Case 1 : $a+c \geq \frac{2n}{\log n}$.} Then either $a\geq \frac{n}{\log n}$ or $c \geq \frac{n}{\log n}$.\\ Firstly, suppose $a\geq \frac{n}{\log n}$, then by Lemma $\ref{lem:rank-upper bound}$, we have
$\rank(M_{f^\varphi})\le 2^{a'' +a'/2 +b +c/2}$. 
Since $2a'+2a''+2b+ c + d = 2n$, we have  $a'/2+a''+b+ c/2 \leq n-a'/2$. By Lemma~\ref{lem:a-type}, $a'\ge 2/5 a \ge 2n/5\log n$. Therefore, $\rank(M_{f\varphi})\le 2^{a'' +a'/2 +b +c/2}  \leq 2^{n-a'/2} \leq  2^{n- \frac{n}{5\log n}}$. \\
Now suppose $c\geq \frac{n}{5\log n}$. Since  $2a'+2a''+2b+ c \leq 2n$, we have $a''+a'+ b +c/2 \le n-c/2 \le n-n/2\log n$. Therefore by Lemma~\ref{lem:rank-upper bound}, 
$\rank(M_{f^\varphi})\le 2^{a'' +a'/2 +b +c/2} \leq 2^{a'' +a' +b +c/2} \leq 2^{n-c/2} \leq 2^{n- \frac{n}{5\log n}}.$
%The last inequality follows from the fact that $\frac{2}{25} < \frac{1}{10}$.
%
%Observe that $\frac{2n}{25\log n} > \frac{2n}{5}$ for $n\geq 2$. Therefore,
%
%$$\rank(M_{F\varphi})\le 2^{\frac{2n}{5}+1}.$$
\\ \noindent\textbf{Case 2 : $a+c < \frac{2n}{\log n}$.} Observe that $b\le n$. Since any type $B$ gate achieves rank 1 under any $\varphi$, by  a simple inductive argument we have $\rank(M_{f^{\varphi}}) \le 2^{a+c+b/2}$ for any  $\varphi$. Therefore 
$\rank(M_{f^\varphi}) \leq 2^{a+c+b/2 } \leq 2^{n/2+2n/\log n} \le 2^{n-n/5\log n}. $
%
%
%Then argue that $$\rank(M_{F^\varphi})\leq 2^{\frac{2n}{5}+1}.$$ Then 
%$$\rank(M_{F^\varphi})\leq 2^{\frac{2n}{5}+1}\leq 2^{n-(\frac{3n}{5}-1)}$$
%
%\textcolor{red}{Now, we compare $\frac{3n}{5}-1$ and $\frac{n^{2/3}}{2}$. Since $\frac{3n}{5}-1 >\frac{n^{2/3}}{2} $ we have $2^{n-\frac{n^{2/3}}{2}}> 2^{n-(\frac{3n}{5}-1)}$. Hence the lower bound.}
%
%
%%There are three cases,
%%
%%{\bf Case~1}: $a\ge n/5$, then by Lemma~\ref{lem:a-type} and \ref{lem:rank-upper bound},
%%$\rank{M_{F\varphi}} \le 2^{a'' +a'/2 +b/2 +c/2} \le 2^{n-a'/2} \le 2^{9/10 n}$ with probability $1-2^{\Omega(n)}$.
%%
%%{\bf Case~2}: $b\ge n/5$, then by Lemma~\ref{lem:rank-upper bound},  $\rank{M_{F\varphi}} \le 2^{a'' +a'/2 +b/2 +c/2} \le 2^{n-n/10}$. 
%%
%%{\bf Case~3}: $a+b \le 2n/5$. That is the number of  gates where both of the children are variables is at most $2n/5$. Removing all the leaves in skew $+$ or  $\times$ gates  and short-circuiting the resulting gates, we conclude that the number of internal gates in $F$ is bounded by  $2n/5$. Hence $F$ has at most $2n/5$ non-skew multiplication gates, and by sub-additivity of the rank function, we conclude 
%%$\rank(M_{F^{\varphi}})\le 2^{2n/5} \le 2^{9n/10}$ for   $\varphi\sim {\cal D}'$.
%\end{proof}
\end{proof}
The following  polynomial  was introduced by Raz and Yehudayoff~\cite{RY08}. 
\begin{defn}
\label{def:raz-poly}
Let $n\in\mathbb{N}$ be an integer. Let $X=\{x_1,\ldots, x_{2n}\}$ and $\mathcal{W} = \{w_{i,k,j} \}_{i,k,j\in[2n]}$. For any two integers $i,j\in\mathbb{N}$, we define an interval $[i,j] = \{ k\in\mathbb{N}, i\leq k\leq j \}$. Let $|[i,j]|$ be the length of the interval $[i,j]$. Let $X_{i,j} = \{ x_p \mid p\in [i,j]\} $ and $W_{i,j}=\{ w_{i',k,j'}\mid i',k,j'\in[i,j] \}$. For every $[i,j]$ such that $|[i,j]|$ is even we define a polynomial $g_{i,j}\in\mathbb{F}[X,\mathcal{W}]$ as
 $g_{i,j}=1$ when   $|[i,j]|=0$  and 
 if $|[i,j]|>0$ then, {\small $g_{i,j }\triangleq (1+x_ix_j)g_{i+1,j-1} + \sum_{k}w_{i,k,j}g_{i,k}g_{k+1,j}.$}
where $x_k$, $w_{i,k,j}$ are distinct variables, $1\le k\le j$ and the summation is over  $k\in [i+1,j-2]$ such that the interval $[i,k]$ is of even length.
Let $g\triangleq g_{1,2n}.$
\end{defn}

In the following, we view $g$ as polynomial in $\{x_1,\ldots, x_{2n}\}$ with coefficients from the rational function field $\mathbb{G} \triangleq \mathbb{F}(\mathcal{W})$.
\begin{lemma}
\label{lem:rank-interval}
Let  Let $X=\{x_1,\ldots, x_{2n}\}$, $Y=\{y_1,\ldots, y_{2n}\}$, $Z=\{z_1,\ldots, z_{2n}\}$ and $\mathcal{W} = \{w_{i,k,j} \}_{i,k,j\in[2n]}$ be sets of variables. Suppose $\varphi\sim\mathcal{D}'$ such that $| |\varphi(X)\cap Y|-|\varphi(X)\cap Z| | = \ell$. Then for the polynomial $g$ as in Definition $\ref{def:raz-poly}$ we have, 
$\rank(M_{g^{\varphi}}) \ge 2^{n-\ell/2}.$
\end{lemma}
\begin{proof}
Proof builds on Lemma~4.3 in~\cite{RY08} as a base case and is by induction on $n+\ell$.

\noindent{\bf Base case:}  Either $\ell = 0 $ or $\ell = 2n$. For $\ell=0$, the statement follows by Lemma~4.3 in~\cite{RY08}. When $\ell=2n$, then  $\rank(M_{g^{\varphi}}) = 1=2^{n-\ell/2}$.

\noindent{\bf Induction step: } Without loss of generality, assume that $|\varphi(X)\cap Y| =  |\varphi(X)\cap Z|+\ell$.  There are three possibilities:
\begin{description}
\item[Case 1 :] Let $\varphi(x_1)\in Y$ and $\varphi(x_{2n})\in Z$ or vice versa. In this case 
\begin{eqnarray*}
\rank(M_{g^{\varphi}}) &\ge& \rank(M_{(1+x_1x_{2n})^\varphi})\rank(M_{g_{2,2n-1}^{\varphi}})= 2\cdot \rank(M_{g_{2,2n-1}^{\varphi}})\\ 
&\ge& 2\cdot 2^{n-1-\ell/2} = 2^{{n-\ell/2}} \hspace*{10 mm}[\text{By Induction Hypothesis.}]
\end{eqnarray*}
\vspace{-4mm}
\item[Case 2 :] $\varphi(x_1)\in Y$ and $\varphi(x_{2n})\in Y$. Then 
\begin{align*}
\rank(M_{g^{\varphi}}) &\ge \rank(M_{(1+x_1x_{2n})^\varphi})\rank(M_{g_{2,2n-1}^{\varphi}}) = 1\cdot \rank(M_{g_{2,2n-1}^{\varphi}})\\
 &\ge 2^{(2n-2)/2- (\ell-2)/2} =2^{n-\ell/2}. \hspace*{10 mm}[\text{By Induction Hypothesis.}]
\end{align*}
For the penultimate inequality above, note that $g_{2,2n-1}$ is defined on $X'=\{x_2,\ldots, x_{2n-1}\}$ and $| |\varphi(X')\cap Y|-|\varphi(X')\cap Z| | = \ell-2$ and hence by Induction Hypothesis, $\rank(M_{g_{2,2n-1}^{\varphi}}) \ge 2^{(2n-2)/2-(\ell-2)/2}$. 
\item[Case 3] $\varphi(x_1)\in Z$ and $\varphi(x_{2n})\in Z$. Then there is an $i\in \{2,2n-1\}$ such that $ | |\varphi(X_i)\cap Y|-|\varphi(X_i)\cap Z| | = 0$ and  $ | |\varphi(X\setminus X_i)\cap Y|-|\varphi(X\setminus X_i)\cap Z| | = \ell$, where $X_i =\{x_1,\ldots, x_i\}$.
Then by the definition of $g$, over $\mathbb{G}$,
$\rank(M_{g^\varphi}) \ge \rank(M_{g_{1,i}^\varphi})\cdot \rank(M_{g_{i+1,2n}^\varphi}) \ge 2^{i/2}\cdot 2^{(2n-i)/2 -\ell/2} =2^{n-\ell/2},
$
since $\rank(M_{g_{1,i}^\varphi}) =2^{i/2}$ by Lemma 4.3 in~\cite{RY08}, and $\rank(M_{g_{i+1,2n}^\varphi})\ge 2^{(2n-i)/2 -\ell/2}$ by Induction Hypothesis.
\end{description}
\vspace{-4mm}
\end{proof}
\begin{lemma}
\label{lem:partition-chernoff}
$\Pr_{\varphi\sim {\cal D}'} [ n -n^{2/3}\le  |\varphi(X)\cap Y| \le n +n^{2/3}] \ge 1- 2^{-\Omega(n^{1/3})}$.
\end{lemma}
\begin{proof}
Proof is a simple application of Chernoff's bound 
(Theorem~\ref{chernoff}) with $\delta = 1/n^{1/3}$. 
\end{proof}

\begin{corollary}
\label{cor:Raz-Yehu}
$\Pr_{\varphi\sim {\cal D}'}[\rank(M_{g^\varphi}) \ge 2^{n-n^{2/3}/2}] \ge 1-2^{-\Omega(n^{1/3})}$.
\end{corollary}
\begin{proof}
Apply Lemma~\ref{lem:rank-interval} with $\ell = n/n^{1/3} =n^{2/3}$ and the probability bound follows from Lemma~\ref{lem:partition-chernoff}.
\end{proof}
\paragraph*{Proof of Theorem~\ref{thm:sum-rop}}
\begin{proof}
Suppose $s< 2^{n/10\log n}$. Then by Lemma~\ref{lem:rank-general} and union bound, probability that there is an $i$ such that $\rank(M_{f_{i}^\varphi}) \ge 2^{n-n/5\log n}$ is $s 2^{-\Omega(n)} =2^{-\Omega(n)}$ and hence by Lemma~\ref{sub-additivity}, $\rank(M_{g^\varphi}) \le s 2^{n-n/5\log n} \le 2^{n-n/10\log n}$ with probability $1-2^{-\Omega(n)}$.  However, by Corollary~\ref{cor:Raz-Yehu}, $\rank(M_{g^\varphi}) \ge 2^{n-n^{2/3}/2}> 2^{n-n/10\log n}$ with probability at least $1-2^{-\Omega(n^{1/3})}$, a contradiction. Therefore, $s=2^{\Omega(n/\log n)}$.  
\end{proof}
%\end{document}

\section{Sum of Products of $\rop$s}
\subsection{$\rop$s under random partition}

Throughout the  section, let   $m\triangleq N^{1/3}$,  $n\triangleq \sqrt{N}$ and  $\kappa = 20\log n$. Let $\cal{D}$  denote the  distribution on the functions $\varphi: X \to Y\cup Z\cup\{0,1\}$ defined as follows
\begin{center}
$\varphi(x_{ij}) \in  \begin{cases} Y &\mbox{with prob. } \frac{m}{N} \\
Z & \mbox{with prob. } \frac{m}{N} \\
1 & \mbox{with prob. } \frac{\kappa n}{N} \\
0 & \mbox{with prob. } 1- \left(\frac{2m+\kappa n}{N}\right) \end{cases}$
\end{center}

\begin{lemma}
\label{lem:monomial-rank}
Let $f$ be an $\rop$ computed by an $\rof$ $F$ and $\varphi\sim {\cal D}$. Let $X$ be a random variable that denotes the number of non-zero multiplication gates at depth $1$. Then 
$\Pr\limits_{\varphi\sim D}\left[ X> \mathcal{O}(N^{1/6}\log n) \right] \leq 2^{-\Omega(m)}$.
\end{lemma}
\begin{proof}
Consider a multiplication gate $g$  at depth 1, with at least two variables as its input. Let $m$ be the monomial (excluding the coefficient) computed by $g$, note that  $d=\deg(m)\geq 2$. we have, $$
\Pr_{\varphi\sim D}[m^\varphi \neq 0] = \left( \frac{2m+\kappa n}{N} \right)^d  \leq \left( \frac{2m+\kappa n}{N} \right)^2 %^ & [\mbox{Since $\frac{2m+\kappa n}{N}<1$}]\\
 \leq \left(\frac{2\kappa n}{N} \right)^2  
\leq \left(\frac{2\kappa}{n}\right)^2  \leq \mathcal{O}\left(\frac{\kappa^2}{N}\right). 
$$
In the above, we have used the fact that $2m<\kappa n$ for large enough $n$.
%Therefore $\Pr\limits_{\varphi\sim D}[m^\varphi \neq 0] \leq \mathcal{O}\left(\frac{\kappa^2}{N}\right) $.
Since $F$ is an $\rof$ in $N$ variables , the $\rop$ $f$ computed by $F$ has at most $N/2$  multiplication gates  where both the inputs are variables.  Then, $\mu \triangleq \mathrm{E}[X] \le \frac{N}{2}\cdot \Pr_{\varphi\sim  D}[m^\varphi \neq 0]\leq N\cdot c\left(\frac{\kappa^2}{N}\right) \leq c(\kappa^2)$, \mbox{ where $c$ is a constant.}
By Theorem~\ref{chernoff}, let $\delta=\frac{N^{1/6}}{\log n}>0$, we have 
$$ \Pr\limits_{\varphi\sim D}\left[ X> \left( 1 +\frac{N^{1/6}}{\log n} \right)c\log^2 n\right ] 
\leq \mathrm{e}^{\frac{-cN^{2/6}}{3}} 
\leq 2^{-\frac{2cN^{1/3}}{3}} \leq 2^{-\Omega(m)}. $$
\end{proof}
\begin{lemma}
\label{lem:remove-monom}
Let $F$ be an $\rof$ computing an $\rop$ $f$ and $\varphi\sim {\cal D}$. Then there exists an $\rof$ $F'$ such that every gate in $F'$ at depth-1 is an addition gate, and 
$\rank(M_{F^\varphi}) \leq \rank(M_{F'^\varphi}) \times 2^{\mathcal{O}(N^{1/6}\log n)}$
with probability atleast $1-2^{-\Omega(m)}$.
\end{lemma}
\begin{proof}
Given an arithmetic formula $F$ we construct the formula $F'$ by replacing every multiplication gate $v$ at depth-1 in $F$ by the constant 1.  Let $X$ of product gates of fan-in  $\ge$1 in $F^\varphi$.  
%Let $u$ be the parent of $v$ in $F$.  If $u$ is a $+$ gate then set $v=0$.
% If $u$ is a $\times$ gate then set $v=1$.
Then, by the construction of $F'$,
$$\rank(M_{F^\varphi}) \leq \rank(M_{F'^\varphi}) \times 2^{X}.$$
Now by Lemma~\ref{lem:monomial-rank}, with probability atleast $1-2^{-\Omega(m)}$ we have,
$$\rank(M_{F^\varphi}) \leq \rank(M_{F'^\varphi}) \times 2^{\mathcal{O}(N^{1/6}\log n)}.$$
\end{proof}

%\begin{proof}
%Given an arithmetic formula $F$ we construct the formula $F'$ by replacing every multiplication gate $v$ at depth-1 in $F$ by the constant 1.  Let $X$ of product gates of fan-in  $\ge$1 in $F^\varphi$.  
%%Let $u$ be the parent of $v$ in $F$.  If $u$ is a $+$ gate then set $v=0$.
%% If $u$ is a $\times$ gate then set $v=1$.
%Then, by the construction of $F'$,
%$\rank(M_{F^\varphi}) \leq \rank(M_{F'^\varphi}) \times 2^{X}.$
%Now by Lemma~\ref{lem:monomial-rank}, with probability atleast $1-2^{-\Omega(m)}$ we have
%$\rank(M_{F^\varphi}) \leq \rank(M_{F'^\varphi}) \times 2^{\mathcal{O}(N^{1/6}\log n)}$.
%\end{proof}

%\begin{defn}
%\begin{defn}
%\label{value-F}
%Let $H$ be a monotone arithmetic formula were the leaves are labeled by numbers in $\mathbb{N}$. Then for any gate $v$ in $H$, the value of $v$ denoted by $\valu(v)$ is  defined  as follows. 
%If $u$ is a leaf then $\valu(u) =a$ where $a\in\mathbb{N}$ is the label of $u$ and $u = u_1 {\sf op} u_2$ then $\valu(u)=\valu(u_1)~{\sf op}~\valu(u_2)$, where ${\sf op}\in \{+,\times\}$.
% Finally, $\valu(H)$ is the value of the  output gate of $H$. 
%\end{defn}

Recall that an arithmetic formula $F$ over $\mathbb{Z}$ is said to be monotone if it does not have any node  labelled by a negative constant. 
%\end{defn}

%\begin{defn}
%\label{value-defn}
%Let $F$ be a binary monotone arithmetic formula with all leaves labeled by either %$1$ or $2$. 
%\end{defn}

%\begin{defn}
%\label{value-F}

%\end{defn}

%Let $w$ be a node in $F$ such that $\valu(w)=0$. \textcolor{red}{Remove $w$ and its sibling from $F^\varphi$.} Now observe that the formula $F^\varphi$ has leaves labelled by $1$ or $2$. 
%Let $F$ be an \rof, and $\varphi:X\to Y\cup Z\cup \{0,1\}$. Suppose without loss of generality, assume that $F^{\varphi}$ is constant minimal (TODO: add this in the preliminaries).

\begin{lemma}
\label{formula-to-monotone}
Let $F$ be an \rof, and $\varphi\sim{\cal D}$. Then there exists a monotone formula $G$ such that $\rank(M_{F^{\varphi}})\le \valu(G^\varphi).$
\end{lemma}
\begin{proof}
Let $F$ be an constant-minimal \rof, and $\varphi\sim {\cal D}$. Let $G^{\varphi}$ be a monotone formula obtained from $F^{\varphi}$ as follows: \\
By short circuiting the gates if necessary,  every leaf node $v$ labelled by a constant is replaced by $1$.   
For every gate $v$ in $F^{\varphi}$ with at least one leaf as a child, 
\begin{itemize}
\item If $v = \prod_{j=1}^k v_j$, with $v_1,\ldots, v_i$, $i\ge 1$ are non-constant leaf gates, then  replace the gates $v_1 \times v_2\times \ldots \times v_i$ by the rank of the polynomial computed by $\varphi(v_1 \times v_2\times \ldots \times  v_i)$.
\item Similarly, if $v = \sum_{j=1}^k v_j$, with $v_1,\ldots, v_i$, $i\ge 1$ are non-constant leaf gates, then  replace the gates $v_1 + v_2+\ldots  + v_i$ by the rank of the polynomial computed by $\varphi(v_1 + v_2 + \ldots + v_i)$.
\end{itemize}
Clearly, the formula constructed above is monotone, since negative constants (if any) in $F^\varphi$ have been replaced by $1$. 
Then, by Lemmas~\ref{sub-additivity} and \ref{sub-multiplicativity}, we have for any $\varphi$, $\rank(M_{F^{\varphi}})\le \valu(G^\varphi)$.
\end{proof}

%\begin{proof}
%Let $F$ be an constant-minimal \rof, and $\varphi\sim {\cal D}$. Let $G^{\varphi}$ be a monotone formula obtained from $F^{\varphi}$ as follows: \\
%By short circuiting the gates if necessary,  every leaf node $v$ labelled by a constant is replaced by $1$.   
%For every gate $v$ in $F^{\varphi}$ with at least one leaf as a child, 
%\begin{itemize}
%\item If $v = \prod_{j=1}^k v_j$, with $v_1,\ldots, v_i$, $i\ge 1$ are non-constant leaf gates, then  replace the gates $v_1 \times v_2\times \ldots \times v_i$ by the rank of the polynomial computed by $\varphi(v_1 \times v_2\times \ldots \times  v_i)$.
%\item Similarly, if $v = \sum_{j=1}^k v_j$, with $v_1,\ldots, v_i$, $i\ge 1$ are non-constant leaf gates, then  replace the gates $v_1 + v_2+\ldots  + v_i$ by the rank of the polynomial computed by $\varphi(v_1 + v_2 + \ldots + v_i)$.
%\end{itemize}
%Clearly, the formula constructed above is monotone, since negative constants (if any) in $F^\varphi$ have been replaced by $1$. 
%Then, by Lemmas~\ref{sub-additivity} and \ref{sub-multiplicativity}, we have for any $\varphi$, $\rank(M_{F^{\varphi}})\le \valu(G^\varphi)$.
%\end{proof}
%

\begin{obs}
\label{obs:rank-value}
 Let $F$ be an $\rof$ and $\varphi\sim {\cal D}$. By Lemma \ref{formula-to-monotone}, we have,
$\Pr[\rank(M_{F^\varphi})>2^r]\leq \Pr[\valu(G^\varphi)>2^r]. $
\end{obs}

\begin{defn}
Let $F$ be an $\rof$ and $\varphi\sim {\cal D}$. A gate $u$ in $F^{\varphi}$ is called a $\sep$, if  either $u$ is a leaf with $\rank(M_{u^{\varphi}}) = 2$, or $u=u_1+u_2$ with $\rank(M_{u_{1}^{\varphi}})=\rank(M_{u_{2}^{\varphi}})=1$ and $\rank(M_{u_{}^{\varphi}})=2$.
\end{defn}

\begin{corollary}
\label{cor:witness}
Let $F$ be an \rof and $\varphi\sim {\cal D}$. Then by Lemma $\ref{lem:rank-2nodes}$ we have
\begin{equation}
\nonumber\resizebox{.9\hsize}{!}{$\Pr[\rank(M_{F^\varphi})>2^r]\leq \Pr[\exists~ u_1,\ldots,u_{\frac{r}{\log N}} \in F^\varphi~ s.t.~ \forall 1\leq i\leq \frac{r}{\log N}~u_i~is~a~\sep ]$}
\end{equation}
\end{corollary}

Now all we need to do is to estimate the probability that a given set of nodes $u_1,\ldots,u_t$ where $t>\frac{r}{\log N}$ are a set of $\sep$s. 

Let $u_i = u_{i,1}+u_{i,2}$ be a $\sep$ in $F^{\varphi}$ and $\rank(M_{u_{i}^\varphi})=2$. Consider the sub-formula rooted at $u_i$. Note that $\rank(M_{u_{i}^\varphi})=2$  only if $\var(u_{i}^\varphi) \cap Y \neq \emptyset $ and  $\var(u_{i}^\varphi) \cap Z \neq \emptyset $. By simple applications of Chernoff's bound, we show that only a small number of $u_1,\ldots, u_t$ can achieve rank-$2$ under a random $\varphi \sim{ \mathcal{D}}$. Let 
$\ell_{i_1},\ldots,\ell_{i_r}$ be the addition gates at depth-1 in the
sub-formula rooted at $u_i$.  For $1\leq i \leq t$ we define $S_i \triangleq \var(\ell_{i_1})\cup \cdots\cup \var(\ell_{i_r})$. Let $v_1,\ldots,v_p$ be the addition gates at depth-1 in $F^\varphi$ that are
not contained in any of the sub-formulas rooted at $u_1,\ldots,u_t$. For $1\leq j \leq p$,  let $S_{t+j}=\var(v_j)$, also let $q = t+p$.

%We have now defined sets $S_1,\ldots,S_q$.

Note that $|S_i| \le s_F \le N^{1/2 + \lambda}$. By merging sets in a greedy fashion whenever  necessary, we assume that $|S_i| \in [N^{1/2 + \lambda}, 2N^{1/2 + \lambda}]$. 
Therefore $q\le N^{1/2 - \lambda}$.

%\begin{claim}
%Let $S_1,\ldots,S_q$ be a collection of sets such that $S_i\subseteq X, 1\leq i \leq q$. Then there exists collection of sets $T_1,\ldots,T_{q'}$ such that for each $1 \leq i \leq q'$, $|T_i|\geq \sqrt{N}$.
%\end{claim}
%\newpage
%\begin{proof}
%The proof is constructive. Consider the following greedy algorithm. 
%\begin{algorithm}
%\caption{Merging variable collections}
%\label{merge}
%\begin{algorithmic}[1]
%\STATE $\mbox{$S = \{S_1,\ldots,S_q\}$}$
%\FOR {$i=1$ to $q$}
%\IF{$|S_i|\geq \sqrt{N}$}
%\STATE $Q=Q\cup S_i$
%\STATE $S=S\setminus S_i$
%\STATE $q=q-1$
%\ENDIF
%\ENDFOR
%\STATE $\mbox{Let the elements of $S$ be $S_1,\ldots,S_q$.}$
%\STATE $R=S_1$
%\STATE $k=2$
%\WHILE{$k<q$} 
%\FOR {$j=k$ to $q$}
%\label{l1}
%\STATE $R=R\cup S_j$
%\IF{$|R|\geq \sqrt{N}$} 
%\STATE $Q=Q\cup R$
%\STATE $break$
%\ELSE 
%\STATE $goto ~\ref{l1}$
%\ENDIF
%\STATE $S=S\setminus \{S_1,\ldots,S_j\}$
%\STATE $k=j+1$
%\STATE $R=\emptyset$
%\ENDFOR
%\ENDWHILE
%\end{algorithmic}
%\end{algorithm}
%
%\textcolor{red}{If the last few sets has size <$\sqrt{N}$ merge it any one of the sets. What if we end up in only one set. }
%
%Therefore given a collection of sets $S_1,\ldots,S_q$ such that $S_i\subseteq X, 1\leq i \leq q$, there exists collection of sets $T_1,\ldots,T_{q'}$ such that for each $1 \leq i \leq q'$, $|T_i|\geq \sqrt{N}$.
%\end{proof} 

For $S\subseteq X$ and $\varphi\sim {\cal D}$ let  $S^\varphi \triangleq\{ \varphi(x)\mid x\in S \}$. Let $W=Y \cup Z$.
We define the following random variables.
\begin{align*}
 X_2 &= \{ S_i\mid 1\leq i \leq q,~ |S_i^\varphi\cap W|= 2 ,~ |S_i^\varphi\cap Y|=1, ~|S_i^\varphi\cap Z|=1 \}.\\
 X_3 &= \{ S_i\mid 1\leq i \leq q,~|S_i^\varphi\cap W|= 3 ,~ |S_i^\varphi\cap Y|\neq \phi, ~|S_i^\varphi\cap Z|\neq \phi \}.\\
X_4 &= \{ S_i\mid 1\leq i \leq q,~|S_i^\varphi\cap W|= 4 ,~ |S_i^\varphi\cap Y|=2, ~|S_i^\varphi\cap Z|=2 \}.\\
X_5 &= \{ S_i\mid 1\leq i \leq q,~|S_i^\varphi\cap W|= 5 ,~ |S_i^\varphi\cap Y|=3 \mbox { and }|S_i^\varphi\cap Z|=2 \mbox{ or vice versa} \}.\\
X_{\geq 6} &= \{ S_i\mid 1\leq i \leq q,~|S_i^\varphi\cap W|\geq 6 ,~ |S_i^\varphi\cap Y|\geq 3, ~|S_i^\varphi\cap Z|\ge 3 \}.
\end{align*}
 Then we have,
\begin{lemma}
\label{lem:chernoff-probabilities}
With the notations as above, 
\begin{itemize}
\item[(1)] $\Pr[ |X_2| + |X_3|+|X_4| +|X_5| \ge 4N^{4/15} ] \le 2^{-\Omega(m)} $; and 
\item[(2)] $(2)\Pr[|X_{\ge 6}| \ge 1  ] \le 2^{-\Omega(m)}  $.
\end{itemize}
\end{lemma}
\begin{proof}
We argue that $\Pr[|X_2| \ge N^{1/5}] \le 2^{-\Omega(m)}$, the argument for the case of $\Pr[|X_3|  \ge N^{1/5}]$, $\Pr[|X_4|  \ge N^{1/5}]$ and $\Pr[|X_5|  \ge N^{1/5}]$ are similar and the result follows by a simple union bound. \\
Let $\mu_2 = \mathbb{E}[|X_2|]= \sum\limits_{i=1}^{q} \frac{|S_i|(|S_i|-1)}{2} \left(\frac{m}{N}\right)^2 \left(1-\frac{m}{N}\right)^{|S_i|-2}$. Since $\lambda\leq \frac{1}{30},q \leq N^{1/2 -\lambda}$ and $|S_i| \in [N^{1/2+\lambda},2N^{1/2+\lambda}]$, we have $\mu_2 = O(N^{1/5})$.
%\mu_2 &= \sum\limits_{i=1}^{q} \frac{|S_i|(|S_i|-1)}{2} \left(\frac{m}{N}\right)^2 \left(1-\frac{m}{N}\right)^{|S_i|-2}
%&\leq  N^{1/2 -\lambda}\frac{|S_i|(|S_i|-1)}{2} \left(\frac{m}{N}\right)^2 \left(1-\frac{m}{N}\right)^{|S_i|-2}\\
%&\leq N^{1/2 -\lambda}2N^{1/2+\lambda} N^{1/2 +\lambda}N^{-4/3}\\
%&\leq 2 N^{1/2 - \lambda 1 + 2\lambda -4/3}
Applying Theorem~\ref{chernoff} with $\delta = \sqrt{\frac{m}{\mu_2}}-1$ we get
$\Pr[|X_2| \ge N^{4/15}] \le 2^{-\Omega(m)}$. 
With a similar argument we get $\Pr[|X_i| \ge N^{4/15}] \le 2^{-\Omega(m)}$ for $i\in\{3,4,5\}$ and (1) follows from union bound.
For (2), we have 
{\small $$\mathbb{E}[|X_{\ge 6}] \le \sum_{i=1}^{q} |S_i|(|S_i|-1)(|S_i|-2)(|S_i|-3)(|S_i|-4) (|S_i|-5)(m/N)^6 (1-m/N)^{|S_i|-6} \le  2^6 N^{-1/2 + 5\lambda}. $$}
Then if $\lambda \leq \frac{1}{30}$, setting $\delta= 1/\mu -1$ in Theorem~\ref{chernoff}, we get    
 $\Pr[|X_{\ge 6}| \ge 1  ] \le 2^{-\Omega(m)}$ as required. \end{proof}

\begin{lemma}
\label{lem:upperbound-sep}
The number of $\sep$ among $u_1,\ldots, u_t$ is at most  $O(N^{4/15})$
with probability at least $1-2^{-\Omega(m)}$.
\end{lemma}
\begin{proof}
Firstly, we  show that  with probability atleast $1-2^{-\Omega(m)}$  among $u_1,\ldots, u_t$ the number of $\sep$s is upper bounded by $|X_2| + |X_3| + 2 (|X_4| + |X_5|),$ which proves the lemma as an immediate consequence.
Note that the sets $X_2, X_3, X_4, X_5$ and $X_{\ge 6}$ are disjoint. Any $S_i\in X_2$ has  exactly one variable each from $Y$ and $Z$, and hence each such $S_i$ can cause at most one of the $u_j$'s to be a \sep. Similarly, $S_i\in X_3$ can also cause at most one of the $u_j$'s to be a \sep. However, an $S_i \in X_4$, can result in at most two of the  gates $u_1,\ldots ,u_q$ being $\sep$s,  since $S_i$ could have been a result of merging two or more linear forms.   Now the bound follows from Lemma~\ref{lem:chernoff-probabilities}.
%Finally, by a simple application of Chernoff bound  $|\im(\varphi)\cap Y| \le 2m$ with probability $1-2^{-\Omega(m)}$, and hence each $S_i \in X_{\ge 6}$ can cause at most $2m$ many of the $u_i$'s  to be rank-$\sep$s.
\end{proof}
%\begin{proof}
%Firstly, we  show that  with probability atleast $1-2^{-\Omega(m)}$  among $u_1,\ldots, u_t$ the number of $\sep$s is upper bounded by $|X_2| + |X_3| + 2 (|X_4| + |X_5|),$ which proves the lemma as an immediate consequence.
%Note that the sets $X_2, X_3, X_4, X_5$ and $X_{\ge 6}$ are disjoint. Any $S_i\in X_2$ has  exactly one variable each from $Y$ and $Z$, and hence each such $S_i$ can cause at most one of the $u_j$'s to be a \sep. Similarly, $S_i\in X_3$ can also cause at most one of the $u_j$'s to be a \sep. However, an $S_i \in X_4$, can result in at most two of the  gates $u_1,\ldots ,u_q$ being $\sep$s,  since $S_i$ could have been a result of merging two or more linear forms.   Now the bound follows from Lemma~\ref{lem:chernoff-probabilities}.
%%Finally, by a simple application of Chernoff bound  $|\im(\varphi)\cap Y| \le 2m$ with probability $1-2^{-\Omega(m)}$, and hence each $S_i \in X_{\ge 6}$ can cause at most $2m$ many of the $u_i$'s  to be rank-$\sep$s.
%\end{proof}

\begin{lemma}
\label{lem:rank-sum-product}
Let $f$ be an \rop on $N$ variables  computed by an \rof~  $F$, with $s_F\le  N^{1/2 + \lambda}$ for some $\lambda \le 1/30$. Then,
$\Pr_{\varphi\sim {\cal D}}[\rank(M_{f^{\varphi}})\ge 2^{N^{4/15}}] \le 2^{-\Omega(m)}.$
\end{lemma}
\begin{proof}
By Corollary~\ref{cor:witness} and,  we have
{\small \begin{eqnarray*}
\Pr[\rank(M_{f^{\varphi}})\ge 2^{N^{4/15}}] &\le& \Pr[\exists~\sep s~ u_1,\ldots ,u_{\frac{N^{1/4}}{\log N}}]\\
&\le &{N \choose \frac{N^{1/4}}{\log N}} 2^{-\Omega(m)} \le 2^{-\Omega(m)}; ~\mbox{ by Lemma~\ref{lem:rank-2nodes} and  since } {N \choose\frac{N^{1/4}}{\log N}} = 2^{o(m)}.
\end{eqnarray*} 
}
\end{proof}

%\begin{proof}
%By Corollary~\ref{cor:witness} and,  we have
%{\small \begin{eqnarray*}
%\Pr[\rank(M_{f^{\varphi}})\ge 2^{N^{4/15}}] &\le& \Pr[\exists~\sep s~ u_1,\ldots ,u_{\frac{N^{1/4}}{\log N}}]\\
%&\le &{N \choose \frac{N^{1/4}}{\log N}} 2^{-\Omega(m)} \le 2^{-\Omega(m)}; ~\mbox{ by Lemma~\ref{lem:rank-2nodes} and  since } {N \choose\frac{N^{1/4}}{\log N}} = 2^{o(m)}.
%\end{eqnarray*} 
%}
%\end{proof}
\subsection{Polynomials with High Rank}
In this section, we prove rank lower bounds for two polynomials under a random partition $\varphi\sim {\cal D}$. 
{The first one is in $\VP$ and the other one is in $\VNP$.}
%\subsubsection*{Product of linear forms}
\begin{lemma}
\label{lem:prod-linear}
Let $p_{lin}=\ell_1\cdots\ell_{m'}$ where {\small  $\ell_j=\left(\sum_{i=(j-1)(N/2m)+1}^{jN/2m} x_i\right)+1$}, where $m'=2m$. Then,
$\rank(M_{{p_{lin}}^\varphi}) = 2^{\Omega(m)}$ with probability $1-2^{-\Omega(m)}$.
\end{lemma}
\begin{proof}
Let $p_{lin}=\ell_1\cdots\ell_{m'}$ where $\ell_j=\left(\sum\limits_{(j-1)(N/2m)+1}^{jN/2m} x_i\right)+1$ and $m'=2m$.

Let us define indicator random variables $\rho_1,\rho_2,\ldots,\rho_{m'}$. 
$$\rho_i =  \begin{cases} 1 &\mbox{if $\rank(M_{\ell_i^\varphi})=2$} \\
0 & \mbox{otherwise } \end{cases}$$
Observe that for any $1 \leq i \leq m', \rank(M_{\ell_i^\varphi})=2$ iff $\ell_{i}^\varphi \cap Y\neq \emptyset \mbox{ and } \ell_{i}^\varphi \cap Z\neq \emptyset$. Therefore, $\Pr[\rank(M_{\ell_i^\varphi})=2] =  \Pr[\ell_{i}^\varphi \cap Y\neq \emptyset \mbox{ and } \ell_{i}^\varphi \cap Z\neq \emptyset]$. For any $1\le j\le m'$, 
$ \Pr[ \ell_{j}^\varphi \cap Y\neq \emptyset \mbox{ and } \ell_{j}^\varphi \cap Z\neq \emptyset] \geq \frac{N}{2m} \left(\frac{N}{2m}-1\right)\left(\frac{m}{N}\right)^2 \left(1-\frac{m}{N}\right)^{\frac{N}{2m}-2}
\ge 1/16~\mbox{ for large enough $N$.} 
$ %\geq \frac{N}{2m} \left(\frac{N}{2m}-1\right)\left(\frac{m}{N}\right)^2 % \left(1-\frac{m}{N}\right)^{\frac{N}{2m}}\\
%&\geq \frac{N}{4m}\left(\frac{N}{2m}-1\right)\left(\frac{m}{N}\right)^2 % \hspace*{10 mm} \textcolor{red}{\text{since $(1-x)^n \approx 1-nx$.}}\\
%&\geq \frac{m}{4N}\left(\frac{N-2m}{2m}\right) \geq \frac{1}{8}\left(1-\frac{2m}{N}\right)\\
%&\approx  \frac{1}{8}\left(1-\frac{1}{N}\right)^{2m}
%\geq \frac{1}{8}\left(1-\frac{1}{N}\right)^{2N}\\
%&\geq \frac{\mathrm{e}^{-2}}{8}
Let $\rho = \sum\limits_{i=1}^{m'} \rho_i$.  Then by linearity of expectation,  $\mu\triangleq \mathbb{E}[\rho] = \sum\limits_{i=1}^{m'} \mathbb{E}[\rho_i] \geq \frac{m}{8}$. By Theorem~\ref{chernoff},  $\Pr[\rho< (1-\delta)\mu] \le e^{-\mu\delta^2/2} =2^{-\Omega(m)}$. Since $\mu \ge m/8$, we have $ \Pr[\rho< (1-\delta)m/8] \le  \Pr[\rho< (1-\delta)\mu] =2^{-\Omega(m)}$. This concludes the proof, by setting $\delta = 1/4$, since $\rank(M_{p_{lin}^\varphi})=2^{\rho}$. 
%\textcolor{red}{
%By Theorem \ref{chernoff}, let $\delta = \frac{1}{3}$
%$\Pr[\rho > \frac{m}{3\sqrt{\mathrm{e}}}] \leq \mathrm{e}^{-\frac{m}{108\sqrt{\mathrm{e}}}}  \leq 2^{-\Omega(m)}$.
%Observe that $\rank(M_{g^\varphi}) =2^{\rho}$. Then $\mathbb{E}[\rank(M_{g^\varphi})] \leq 2^{\mathbb{E}[\rho]}$.
%$$ \Pr[\rank(M_{g^\varphi}) =2^{\rho} ] $$ 
% }
%\textcolor{red}{\textbf{DON'T WE NEED A UPPER BOUND ON $\mathbb{\rho}$} to apply chernoff 3rd inequality ?}
\end{proof}

Throughout the section let $\varphi$ denote a function of the form  $\varphi: X\to Y\cup Z\cup \{0,1\}$.  Let $X_\varphi$ denote the matrix $(\varphi(x_{ij}))_{1\le i,j\le n}$. If and when $\varphi$ involved in a probability argument, we assume that $\varphi$ is distributed according to $\mathcal{D}$.

\begin{defn} Let $1\le i,j\le n$.  $(i,j)$ is said to be a {\em Y-special } $($respectively {\em Z-special}$)$ if
%\begin{center}
%\textcolor{red}{$\varphi(x_{ij})\in Y~(\text{respectively~}\varphi(x_{ij})\in Z)$}\\
 $\varphi(x_{ij})\in Y~(\text{respectively~}\varphi(x_{ij})\in Z)$, 
$\forall i'\in[n],i'\neq i~~\varphi(x_{i'j})\in\{0,1\}$ and 
$\forall j'\in[n],j'\neq j~~\varphi(x_{ij'})\in\{0,1\}$.
%\end{center}
\end{defn}

%Similarly we define Z-special position
%
%\begin{defn}
%An index $(i,j)$ in an $n\times n$ matrix is said to be a {\em Z-special position} if and only if $\varphi(x_{ij})\in Z$ and every other entry in the $i^{th}$ row and $j^{th}$ column is in $\{0,1\}$. That is, 
%\begin{center}
%$\forall i'\in[n],i'\neq i~~\varphi(x_{i'j})\in\{0,1\}$
%\\
%$\forall j'\in[n],j'\neq j~~\varphi(x_{ij'})\in\{0,1\}$
%\end{center}
%\end{defn}

%\subsection*{Expected number of Y variables as per $D$

\begin{lemma}
\label{Ybound}
Let $\Q\in \{Y,Z\}$, $\varphi$ as above and  $\chi = | \varphi(X)\cap \Q|$ where $\varphi(X)=\{\varphi(x_{ij})\}_{i,j\in[n]}$. Then,
$\Pr\limits_{\varphi\sim\cal{D}}\left[\frac{3m}{4} < \chi < \frac{5m}{4}\right] > 1-2^{-\Omega(m)}$.
\end{lemma}
\begin{proof}
Define indicator random variables $\chi_{ij}$ for $1\leq i,j\leq n$:
$$\chi_{ij} = \begin{cases} 1 &\mbox{if $\varphi(x_{ij})\in \Q$}\\
0 & \mbox{otherwise}. \end{cases}$$  

Then $\chi = \sum_{i=1}^n\sum_{j=1}^n\chi_{ij}$ and $\mathop{{}\mathbb{E}}_{\varphi\sim{\mathcal{D}}}[\chi]=m$. %We have
%\begin{align*}
%\mathbb{E}[\chi] &= \sum\limits_{i=1}^n\sum\limits_{j=1}^n\mathbb{E}[\chi_{ij}] =  \sum\limits_{i=1}^n\sum\limits_{j=1}^n\Pr[\varphi(x_{ij})\in Y]=\sum\limits_{i=1}^n\sum\limits_{j=1}^n \frac{m}{N}\\
%&= n^2 \cdot \frac{m}{N} = m\\
%\mu  = \mathbb{E}[\chi] &= m &&[\text{Since~} N=n^2]
%\end{align*}
Let $\delta=\frac{1}{4}$, 
then by Chernoff bounds in Theorem~\ref{chernoff},
$$\Pr\left[\chi \geq \frac{5m}{4} \right] \leq \mathrm{e}^{-\frac{\delta^2\mu}{3}} \leq \mathrm{e}^{-\frac{m}{48}}  = 2^{-\Omega(m)}; \mbox{ and}
\Pr\left[\chi \leq \frac{3m}{4} \right] \leq \mathrm{e}^{-\frac{\delta^2\mu}{2}} \leq \mathrm{e}^{-\frac{m}{32}}=  2^{-\Omega(m)}.$$ Therefore, $\Pr\limits_{\varphi\sim\cal{D}}\left[\frac{3m}{4} < \chi < \frac{5m}{4}\right] = 1-2^{-\Omega(m)}.$
\end{proof}

Let $C_1,\ldots, C_n$ denote the columns of $X_\varphi$ and $R_1,\ldots, R_n$ denote the rows of $X_{\varphi}$.
\begin{defn} Let $\Q\in \{Y,Z\}$. A column $C_j$, $1\le j\le n$ is said to be {\em $\Q$-good} if 
$\exists i\in[n], ~~\varphi(x_{ij})\in \Q$; and 
$\forall i'\in[n],i'\neq i ~~\varphi(x_{i'j})\in\{0,1\}$.
{\em $\Q$-good} rows are defined analogously.
\end{defn}

%\begin{defn}Let $C_1,C_2,\ldots,C_n$ be columns of $A_\varphi$. We say $C_j$ is a {\em Z-good column} if and only if there exists exaclty one entry in $C_j$ is in $Z$ and every other entry in $C_j$ is in $\{0,1\}$. That is,
%\begin{center}
%$\exists i\in[n], ~~\varphi(x_{ij})\in Z$\\
%$\forall i'\in[n],i'\neq i ~~\varphi(x_{i'j})\in\{0,1\}$
%\end{center}
%\end{defn}
%\begin{defn}Let $R_1,R_2,\ldots,R_n$ be columns of $A_\varphi$. We say $R_i$ is a {\em Z-good row} if and only if there exists exaclty one entry in $R_i$ is in $Z$ and every other entry in $R_i$ is in $\{0,1\}$. That is,
%\begin{center}
%$\exists j\in[n], ~~\varphi(x_{ij})\in Z$\\
%$\forall j'\in[n],j'\neq j ~~\varphi(x_{ij'})\in\{0,1\}$
%\end{center}
%\end{defn}

%\begin{obs}For $1\leq i,j\leq n$, let $C_i,C_j$ be two columns~$($respectively rows$)$  in $X_\varphi$. Let $P$ be the event that $C_i$ is a Y-good column~$($respectively a {\em Y-good row}$)$ and $Q$ be the event that $C_j$ is a Y-good column~$($respectively a {\em Y-good row}$)$. Then $P$ and $Q$ are independent events. 
%\end{obs}

\begin{obs}
\label{obs:mutex} 
Let $C_i$ be a Y-good column in $X_\varphi$. Let $i,i'\in[n]$, $R$ be the event that $\varphi(x_{ij})\in Y$ and $T$ be the event that $\varphi(x_{i'j})\in Y$. The events $R$ and $T$ are mutually exclusive. 
\end{obs}
 
 By Observation~\ref{obs:mutex} and union bound we have:
\begin{lemma} 
\label{lem:ygood}
For $1\leq i\leq n$, let $C_i$ be a column in $X_\varphi$. Then for any $\Q\in\{Y,Z\}$, 
$\Pr\limits_{\varphi\sim\cal{D}}[\text{$C_i$ is $\Q$-good}] = n\cdot\frac{m}{N}\left(1-\frac{2m}{N}\right)^{n-1}.$
\end{lemma} 
For $\Q\in \{Y,Z\}$  let 
$\eta_{\Q} \triangleq|\{C_i~|~C_i \mbox{ is $\Q$-good}\}|$  \mbox{ and}   $
\zeta_{\Q} \triangleq|\{R_j~|~R_j \mbox{ is $\Q$-good}\}
$.
\begin{lemma} 
\label{Ygoodbound}
With notations as above,  $\forall \Q\in\{Y,Z\}$,
$\Pr\limits_{\varphi\sim \mathcal{D}}[\eta_{\Q} \ge \frac{2m}{3}] \ge 1 - \frac{1}{2^{\Omega(m)}}$;
and $\Pr\limits_{\varphi\sim \mathcal{D}}[\zeta_{\Q} \ge \frac{2m}{3}] \ge 1 - \frac{1}{2^{\Omega(m)}}$.
\end{lemma} 
\begin{proof}
Proof is a simple application for Chernoff's bound. We argue for the case of $\eta_Y$, the rest are analogous. For $1 \leq i \leq n$, let
%\begin{center}
$$\eta_{i} = \begin{cases} 1 &\mbox{if $C_i$ is Y-good column}\\
0 & \mbox{otherwise}. \end{cases}$$

%\end{center}
Then $\eta_Y = \eta_1 + \cdots + \eta_n$ and by Observation~\ref{obs:mutex} and Lemma~\ref{lem:ygood} $
\mathbb{E}[\eta_i] = \Pr[\text{$C_i$ is Y-good}]
= n\cdot\frac{m}{N}\left(1-\frac{2m}{N}\right)^{n-1} $. 
By linearity of expectation,
$
\mathbb{E}[\eta_Y] = n^2\cdot\frac{m}{N}\left(1-\frac{2m}{N}\right)^{n-1} = m\left(1-\frac{2m}{N}\right)^{n-1}$.

Set $\rho= \left(1-\frac{2m}{N}\right)^{n-1}$ so that $\mathbb{E}[\eta_Y]=\rho m$. 
For $\delta = \frac{1}{4}$, we have by Theorem \ref{chernoff},
$$\Pr\left[\eta_Y\leq \left(1-\frac{1}{4}\right)\rho m\right] \leq \mathrm{e}^{\frac{- (1/4)^2\mu}{2}} \leq \mathrm{e}^{-\mu/32}.$$
As $m=o(n)$ and $N=n^2$, $\lim\limits_{n\to \infty}\frac{2m}{N}= 0$. Thus for sufficiently large $n$, $\rho\ge  9/10$ and hence $\mu \ge 9m/10$.  We conclude $
Pr\left[\eta_Y \leq 27m/40\right] \le 2^{-\Omega(m)}.$  Since $27/40>2/3$ we have $\Pr[\eta_Y \ge \frac{2m}{3}] \ge 1 - \frac{1}{2^{\Omega(m)}}$ as required.  
\end{proof}
%\begin{proof}
%Proof is a simple application for Chernoff's bound. We argue for the case of $\eta_Y$, the rest are analogous. For $1 \leq i \leq n$, let
%%\begin{center}
%$\eta_{i} = \begin{cases} 1 &\mbox{if $C_i$ is Y-good column}\\
%0 & \mbox{otherwise}. \end{cases}$
%
%%\end{center}
%Then $\eta_Y = \eta_1 + \cdots + \eta_n$ and by Observation~\ref{obs:mutex} and Lemma~\ref{lem:ygood} $
%\mathbb{E}[\eta_i] = \Pr[\text{$C_i$ is Y-good}]
%= n\cdot\frac{m}{N}\left(1-\frac{2m}{N}\right)^{n-1} $. 
%By linearity of expectation,
%$
%\mathbb{E}[\eta_Y] = n^2\cdot\frac{m}{N}\left(1-\frac{2m}{N}\right)^{n-1} = m\left(1-\frac{2m}{N}\right)^{n-1}$.
%
%Set $\rho= \left(1-\frac{2m}{N}\right)^{n-1}$ so that $\mathbb{E}[\eta_Y]=\rho m$. 
%For $\delta = \frac{1}{4}$, we have by Theorem \ref{chernoff},
%$
%\Pr\left[\eta_Y\leq \left(1-\frac{1}{4}\right)\rho m\right] \leq \mathrm{e}^{\frac{- (1/4)^2\mu}{2}} \leq \mathrm{e}^{-\mu/32}.
%$
%As $m=o(n)$ and $N=n^2$, $\lim\limits_{n\to \infty}\frac{2m}{N}= 0$. Thus for sufficiently large $n$, $\rho\ge  9/10$ and hence $\mu \ge 9m/10$.  We conclude $
%Pr\left[\eta_Y \leq 27m/40\right] \le 2^{-\Omega(m)}.$  Since $27/40>2/3$ we have $\Pr[\eta_Y \ge \frac{2m}{3}] \ge 1 - \frac{1}{2^{\Omega(m)}}$ as required.  
%\end{proof}

\begin{lemma}
\label{lem:special-pos}
For ${\cal Q}\in \{Y,Z\}$, let $\gamma_ \Q$ denote the number of $\Q$-special positions  in $X_\varphi$.  Then $\forall \Q\in \{Y, Z\}$,
$\Pr\limits_{\varphi\sim\cal{D}}\left[\gamma_\Q \geq \frac{m}{12}\right] \geq 1-2^{-\Omega(m)}.$
\end{lemma}
\begin{proof} We argue for $\Q=Y$, the proof is analogous when $\Q=Z$.
Let $\varphi$ be distributed according to $\mathcal{D}$.  Consider the following events  on $X_\varphi$.  
$\mbox{E1} :  2m/3 \le |X_\varphi \cap Y| \le 5m/4$;
$\mbox{E2} : \mbox{The number of $Y$-good columns  and $Y$-good rows  is at least $r\triangleq 2m/3$.}$

By Lemmas~\ref{Ybound} and~\ref{Ygoodbound}, $X_\varphi$ satisfies the events E1 and E2 with probability $1-2^{-\Omega(m)}$.   Henceforth we assume that $X_{\varphi}$ satisfies the events E1 and E2.  
 
Without loss of generality, let $R_1,\ldots,R_r$ be  the first $r$ $Y$-good rows in $X_\varphi$. For every $Y$-good row $R_i$, $1\le i\le r$ there exists a corresponding witness column $C_j,j\in[n]$ such that $\varphi(x_{ij})\in Y$. Without loss of generality, assume  $C_1,\ldots,C_r$ be  columns that are  witnesses for  $R_1,\ldots, R_r$ being $Y$-good. Further, $X_\varphi(C_j)$  denote the set of values along the column $C_j$.  %There are two possibilities:

%\begin{case}For all $j\in[r]$, $C_j$ is a Y-good column
%\end{case}By definition, if $R_i$ is a Y-good row and $C_j$ is a Y-good column then $(i,j)$ is a Y-special position. Therefore if $C_j$ is a Y-good column for all $j\in[r]$, then the number of Y-special positions is $r$.
Suppose   among $C_1,\ldots, C_r$,  $t\ge 0$ columns are not Y-good, without loss of generality let them be  $C_1,C_2,\ldots,C_t$. %Let $p_j=|X_{\varphi}(C_j) \cap Y|$ and  $q_j=|X_{\varphi}(C_j) \cap Z|$. 

%If $p\geq 2$ or $q \geq 2$, then the number of $Y$ variables in columns that are not good is atleast $2t$. By Lemma \ref{Ybound}, the number of $Y$ variables is atmost $\frac{5m}{4}$ with probability $1-2^{-m}$ and by Lemma \ref{Ygoodbound}, the number of $Y$ variables in Y-good columns is atleast $\frac{2m}{3}$ with probability $1-2^{-m}$. Then,
%\begin{align*}
%2t \leq \frac{5m}{4} - \frac{2m}{3}\\
%t \leq \frac{7m}{24}
%\end{align*}
Each of the column $C_j$ has at least  one variable from  $Y$  and hence the columns $C_1,\ldots, C_t$ contain at least $t$ distinct variables from $Y$.  By  event E2, there are at least $\frac{2m}{3}$ $Y$-good columns that are distinct from $C_1,\ldots, C_t$, each containing exactly one distinct variable from $Y$. Since the total number of variables from $Y$ in $X_{\varphi}$ is at most $5m/4$ (by E1) we have,
%\begin{align*}
$t \leq \frac{5m}{4} - \frac{2m}{3}\leq \frac{7m}{12}.$
%\end{align*}
That is, at most $7m/12$ of the columns among  $C_1,\ldots, C_r$ are not $Y$-good.  Therefore, at least $r-t$ of the columns among $C_1,\ldots, C_r$ are $Y$ good and hence the number of  $Y$-special positions in $X_\varphi$ is atleast $r-t \ge (2/3-7/12)m = \frac{m}{12}$.  %Since $X_\varphi$ satisfies the conditions 1 \& 2 with probability at least $1-2^{\Omega(m)}$ 
We conclude,
%\begin{center}
$\Pr\limits_{\varphi\sim\cal{D}}\left[\gamma_Y \geq \frac{m}{12}\right] \geq 1-2^{-\Omega(m)}$.
%\end{center} 
\end{proof}

A  row  $R$ in the matrix $A \in (Y\cup Z\cup\{0,1\})^{n\times n}$ said to be {\em $1$-good} if there  is at least one $1$ in $R$ in a column  other than    $Y$-special and $Z$-special positions. The following observation is immediate :
\begin{obs}
\label{lem:ones}
Let $\varphi$ be distributed according to $\mathcal{D}$. Then for any row $($column$)$ $R$: 
$\Pr\limits_{\varphi\sim\cal{D}}\left[\mbox{R is $1$-good}\right]\ge (1-1/n^3).$
\end{obs}

%\paragraph*{Permanent has high rank}
 Finally, we are ready to show that $\perm$ has high rank under a random $\varphi\sim \mathcal{D}$.
 
\begin{theorem}
\label{thm:perm-rank}
$\Pr[\rank(M_{{\perm}_{n}^\varphi}) \ge2^{m/12}] \ge (1-O(1/n^2))/2$.
\end{theorem}
 We need a few notations and Lemmas before proving Theorem~\ref{thm:perm-rank}.  
Consider a $\varphi:X\to Y\cup Z\cup \{0,1\}$ and let the number of $Y$-special positions and the number of $Z$-special positions in $X_{\varphi}$ are both be at least $\gamma$. Let $(i_1,j_1),(i_2,j_2),\ldots,(i_\gamma,j_\gamma)$ be a set of distinct $Y$- special positions that do not share any row or column and $(k_1,\ell_1),(k_2,\ell_2),\ldots,(k_\gamma,\ell_\gamma)$ be a set of distinct $Z$ - special positions in $X_\varphi$ that do not share any row or column.    

%Let $G=(A\cup B,E )$ be a bipartite graph where $A=\{ (i_1,j_1),(i_2,j_2),\ldots,(i_\gamma,j_\gamma)\}$ and $B=\{ (k_1,\ell_1),(k_2,\ell_2),\ldots,(k_\gamma,\ell_\gamma)\}$.
%
%\begin{defn}
%A matching $\mathcal{M}$ in $G$ is said to be a {\em lexicographically ordered} if 
%\end{defn}
%
%\begin{defn}A matching $\mathcal{M}$ in $G$ is said to be a {\em lex-least matching} if it has the lowest lexicographical order among the set of all matchings in $G$.
%\end{defn}
Without loss of generality, suppose $i_1<i_2< \cdots <i_\gamma$ and $k_1<k_2<\cdots < k_\gamma$.  Let $\mathcal{M}$ be the  perfect matching $((i_1,j_1),(k_1,\ell_1)) , \ldots, ((i_\gamma, j_{\gamma}),(k_{\gamma},\ell_{\gamma}))$. \\
For an edge $\{(i_p,j_p),(k_p,\ell_p)\}\in\mathcal{M}$, $1\le p\le \gamma$ consider the  $2\times 2$ matrix :
\begin{center}
$B_p =\begin{pmatrix}
   X_\varphi[i_p,j_p] &  X_\varphi[i_p,\ell_p] \\
    X_\varphi[k_p,j_p] &  X_\varphi[k_p,\ell_p]
\end{pmatrix}.$
\end{center}

There exists a partition $\varphi: X \rightarrow Y \cup Z \cup \{0,1\}$ such that $\rank(M_{B_p^\varphi})=2$. Let $A$ be the matrix obtained by permuting the rows and columns in $X_\varphi$ such that $A$ can be written as in the  Figure~\ref{fig:figure1} below.
\begin{figure}[h]
\begin{center}
\epsfig{file=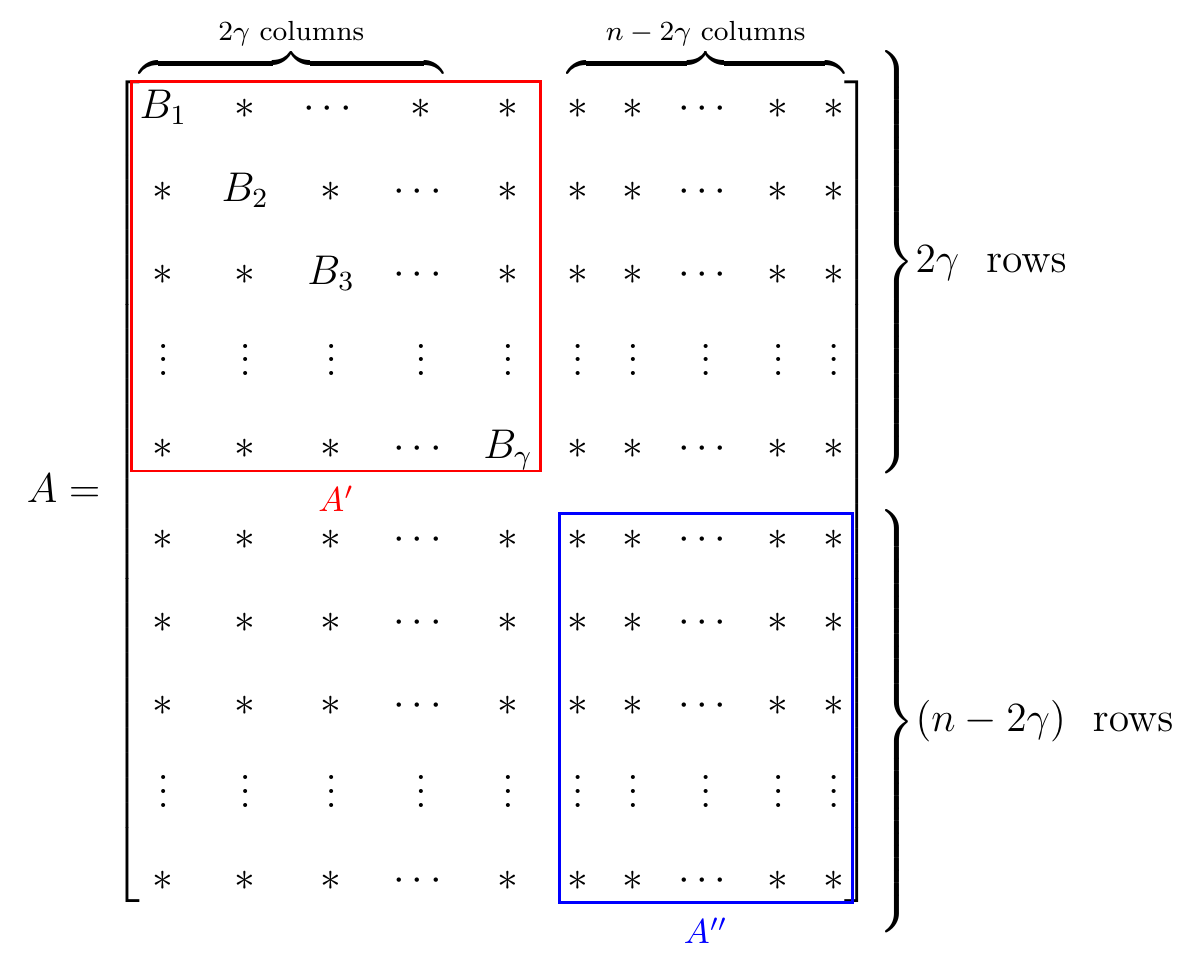, scale=.45}
%\hspace*{20 mm}\includegraphics[scale=0.6]{finalpicture.pdf}
\caption{The matrix $A$ after permuting the rows an columns. $*$ denotes unspecified entries.}
\label{fig:figure1}
\end{center}
\end{figure}
%\vspace{-6mm}
%where $*\in\{0,1\}$. 
Since $(i_p,j_p)$ is a $Y$-special position, $(k_p,\ell_p)$ is a $Z$-special position we have $X_\varphi[i_p,j_p]\in Y$, $X_
\varphi[k_p,\ell_p]\in Z$. Also $X_\varphi[i_p,\ell_p]
\in\{0,1\}$ and $X_\varphi[k_p,j_p]\in\{0,1\}$. 
Further, $\rank(M_{{\perm}(B_p)})=2$ if and only if  $ 
X_\varphi[k_p,j_p]=X_\varphi[i_p,\ell_p]=1$.         
Consider the following events:
$F_1$: $\gamma\ge m/12$;~~ and
$F_2$: Rows $i_1,\ldots, i_\gamma, k_1\ldots,k_\gamma$ are $1$-good.
\begin{lemma}
\label{perm-nonzer-rank}
Let ${A}$  and $A''$ be matrices as above. Then 
%\begin{align*}
$\Pr_{\varphi}[perm({A''})\neq 0~|~F_1,F_2]\geq 1-\frac{1}{n^2}.$
%\end{align*}
\end{lemma}
\begin{proof}
Permanent of any matrix $M$  with entries from $Y\cup Z\cup \{0,1\}$ is zero if and only if $M$ has an all zero $s\times t$ sub matrix such that $s+t=n+1$.(See Theorem~12.1 in \cite{vanL}.) We begin with a bound on the probability that there is at least one  column/row with all zero entries.  Note that the event $F_1$ depends only on the entries of $X_\varphi$ being in $Y\cup Z$,  and the event  $F_2$ is independent of the rows and columns of $A''$.
Thus, for any position $(i,j)$ in $A''$, we have $\Pr[\varphi(x_{i,j})=1]= \kappa n/N(1-2m/N) \approx \kappa n/N$, for large enough $n$. Thus,
\begin{align*}
\Pr[\forall j\in[n], \varphi(x_{ij})=0 ] &\le \left(1-\frac{\kappa n}{N} \right)^{n-2\gamma}\mbox{ and hence, }\\
\Pr[\exists i~\forall j\in[n], \varphi(x_{ij})=0 ] &\leq n\cdot \left( 1-\frac{\kappa n}{N} \right)^{n-2\gamma}~~\mbox{ by union bound}\\
%&\leq n\cdot \left( 1-\frac{\kappa}{n} \right)^{n-2\gamma}
\mbox{Since $\gamma=\mathcal{O}(m)=o(n)$ and $N=n^2$,}\\
\Pr[\exists i~\forall j\in[n], \varphi(x_{ij})=0 ] &\leq n \frac{\left( 1-\frac{\kappa}{n} \right)^n}{\left( 1-\frac{\kappa}{n} \right)^{2\gamma}} %\mbox{  Since $\gamma=\mathcal{O}(m)=o(n)$ and $N=n^2$.}
\end{align*}
As $n\rightarrow\infty$, the denominator $\left(1-\frac{\kappa}{n}\right)^{2\gamma}\rightarrow 1$.
Now, consider $1<c<n'-1$, where $n'=n-2\gamma$. We estimate the probability that there exists an $c\times (n'-c+1)$ all zero sub-matrix of  $A''$. For any $c\times (n'-c+1)$ sub-matrix $M$ of $A''$, $\Pr[M=0] = (1-\kappa/n)^{c(n'-c+1)}$.

As there are ${n' \choose c}^2$ many such sub-matrices $M$ of $A''$, we get
\vspace{-3mm}
\begin{align*}
\Pr[\exists M, M=0] &\le {n' \choose c}^2 (1-\kappa/n)^{c(n'-c+1)}\\
&\le (n'e/c)^c (1-\kappa/n)^{c(n'-c+1)} \approx e^{2c\log((n+1)/c)-\kappa c(n'-c+1)/n} \le  e^{- 4 \log n}
\end{align*}
the last inequality follows since, $\kappa =20\log n$, and  hence $2c \log(n+1/c)-\kappa (c-1)(n'-c+1)/n\le -2 $ for large enough $n$.
{\small  
\begin{align*}
\Pr[perm(\mathcal{A''})=0\mid F_1,F_2] &\leq n\cdot \left( 1-\frac{\kappa}{n} \right)^n + n e^{- 4 \log n} \leq n \left[ \left( 1-\frac{\kappa}{n} \right)^{n/\kappa} \right]^\kappa
 + 1/n^3 \leq n\cdot \mathrm{e}^{-\kappa}\leq  1/n^2.
\end{align*}}
The penultimate inequality in the above is obtained by substituting $\kappa=20\log n$.
\end{proof}

Let $F_3$ denote the event ``${\perm}(A'')\neq 0$''.
Define sets of matrices:
\begin{align*}
\mathcal{A}&\stackrel{\triangle}{=} \left\{ X_\varphi\mid \parbox{1.7in}{ $X_\varphi \in F_1\cap F_2\cap F_3$ and    
 $\exists i\le \gamma, \rank(M_{{\perm}(B_i)})=1$ } \right\};&
\mathcal{B} \stackrel{\triangle}{=}\left\{ X_\varphi\mid\parbox{1.7in}{ $X_\varphi\in F_1\cap F_2 \cap F_3$ and 
 $\forall i\le \gamma, \rank(M_{{\perm}(B_i)})=2$. } \right\} 
\end{align*}

\begin{obs}
\label{obs:rank-1}
$\forall A\in \mathcal{A}$, $\rank({\perm}(A'))< 2^{\gamma}$ and $\forall B\in \mathcal{B}, \rank({\perm}(B))\ge 2^{\gamma}$.
\end{obs}
\begin{lemma}
\label{lem:density} Let $\mathcal{A}$ and $\mathcal{B}$ as defined above. Then
(a) $\Pr\limits_{\varphi\sim \mathcal{D}}[\rank(M_{{\perm}(X_{\varphi})})\ge 2^{\gamma}
)] \ge {\cal D}({\cal B})$; and
(b)  ${\cal D}(B)\ge {\cal D}({\cal A})$,
 where  ${\mathcal D}(S)= \Pr\limits_{\varphi\sim \mathcal{D}}[X_\varphi \in S]$ for  $S\in \{\mathcal{A}, \mathcal{B}\}$.
\end{lemma}
\begin{proof}
(a) follows from Observation~\ref{obs:rank-1}. For (b), we establish a one-one mapping $\pi: \mathcal{A} \to \mathcal{B}$ defined as follows.
Let $\varphi$ be such that $X_\varphi\in \mathcal{A}$.  Consider $1\le p\le \gamma$ such that $\rank(M_{{\perm}(B_p)})=1$.  Then either $X_{\varphi}[k_p,j_p]=0$ or $X_{\varphi}[i_p,\ell_p]=0 $ or both. If  $X_{\varphi}[k_p,j_p]=0$, then set $X_{\varphi'}[k_p,j_p]=1$, and $X_{\varphi'}[k_p,\iota_p]=0$ where $\iota_p\in[n]\setminus\{j_1\ldots, j_\gamma, \ell_1\ldots, \ell_\gamma\}$ is the first index from left such that  $X_{\varphi}[k_p,\iota_p]=1$. Similarly, if  $X_{\varphi}[i_p,\ell_p]=0$, then set $X_{\varphi'}[i_p,\ell_p]=1$, and $X_{\varphi'}[i_p,\lambda_p]=0$ where $\lambda_p\in[n]\setminus\{j_1\ldots, j_\gamma, \ell_1\ldots, \ell_\gamma\}$
is the first index from left such that  $X_{\varphi}[k_p,\lambda_p]=1$. 
Let $\varphi'$ be the partition obtained from $\varphi$ by applying the above mentioned swap operation for every $1\le p\le \gamma $ with  $\rank(M_{{\perm}(B_p)})=1$, keeping other values of $\varphi$ untouched. Clearly $X_{\varphi'}\in \mathcal{B}$.  Set $\pi(X_\varphi) \mapsto X_{\varphi'}$.   It can be seen that $\pi$ is an one-one map. Further, for any fixed $A\in \mathcal{A}$,  $\Pr_\varphi[X_\varphi=A] = \Pr_\varphi[X_{\varphi} =\pi(A)]$ since $\varphi$ is independently and identically distributed for any position in the matrix. Thus we have $\mathcal{D}(\mathcal{A})\le\mathcal{D}(\mathcal{B})$.
\end{proof}

\begin{proof}[Proof of Theorem~\ref{thm:perm-rank}]
It is enough to argue that $\Pr_{\varphi\sim \mathcal{D}}[X\varphi\in \mathcal{A}\cup\mathcal{B}] = 1-O(\frac{1}{n^2})$, as $\mathcal{A}\cap\mathcal{B}=\emptyset$ . 
Now, $\Pr_{\varphi\sim \mathcal{D}}[X_\varphi\in \mathcal{A}\cup\mathcal{B}] =\Pr_{\varphi\sim \mathcal{D}}[F_1\cap F_2\cap F_3] $.  By Lemma~\ref{Ybound}, $Pr_{\varphi\sim \mathcal{D}}[F_1] = 1-2^{-\Omega(m)}$, by Lemma~\ref{lem:ones} combined with union bound we have 
$Pr_{\varphi\sim \mathcal{D}}[F_2] \ge 1-\gamma/n^3$ and by Lemma~\ref{perm-nonzer-rank} $Pr_{\varphi\sim \mathcal{D}}[F_3|F_1,F_2] \ge 1-2/n^2$.
Thus we conclude $\Pr_{\varphi\sim \mathcal{D}}[F_1\cap F_2\cap F_3] = 1-O(\frac{1}{n^2})$.
As $\mathcal{D}(\mathcal{B}\cap\mathcal{A})= \mathcal{D}(\mathcal{A})+\mathcal{D}(\mathcal{B})
$, by Lemma~\ref{lem:density} we have  $\Pr_{\varphi\sim \mathcal{D}}[\rank(M_{{\perm}(X_{\varphi})}) \ge 2^{\gamma}] \ge 1/2(1-O(\frac{1}{n^2}))$.
\end{proof}
\subsection{Putting them all together}
\label{sec:final}
%In this section we combine the above observations to  prove our main theorems. 
\paragraph*{Proof of Theorem~\ref{thm:depth-three}}
\begin{proof}
Suppose $p_{lin}= \sum_{i=1}^s\prod_{j=1}^t f_{i,j}$ where $f_{i,j}$ are syntactically multi-linear $\Sigma\Pi\Sigma$ formula, with $s<2^{N^{1/4}}$, Let $f_{i,j}= \sum_{k=1}^{s'} T_{i,j,k}$, and $T_{i,j,k}$ are products of variable disjoint linear forms, and hence $\rop$s. Further, since the bottom fan-in of each $f_{i,j}$ is bounded by $N^{1/2+\lambda}$, we have $s_{T_{i,j,k}}\le 2^{N^{1/2+\lambda}}$.
Then by Lemma~\ref{lem:rank-sum-product} and union bound there is an $i,j,k$ such that  $\rank(M_{T_{i,j,k}^\varphi}) \ge 2^{N^{4/15}}$ with probability at most  $s t s'2^{-\Omega(m)}$.
 By Lemma~\ref{sub-aditivity2} and \ref{sub-multiplicativity2}, we have   $\mrank(\pcm_{p_{lin}^{\varphi}}) \le 2^{N^{4/15}}$  with probability $1-o(1)$.  However  by Lemma~\ref{lem:prod-linear}, $\mrank(\pcm_{p_{lin}^\varphi})=\rank(M_{p_{lin}^\varphi}) = 2^{\Omega(m)}$ with probability at least $1-2^{-\Omega(m)}$, a contradiction. Hence $ss'= 2^{\Omega(N^{1/4})}$. 
\end{proof}

\paragraph*{Proof of Theorem~\ref{thm:sum-product-rop}}
\begin{proof}
Suppose $s < 2^{N^{1/4}}$. Then by Lemma~\ref{lem:rank-sum-product}, the probability that there is an $f_{i,j}$ with $\rank(M_{f_{i,j}^\varphi}) \ge 2^{N^{4/15}}$ is at most $2^{-\Omega(m)} s\cdot t =o(1)$. By Lemma~\ref{sub-aditivity2} and \ref{sub-multiplicativity2}, we have   $\mrank(\pcm_{p_{lin}^\varphi}) \le 2^{N^{4/15}}$  with probability $1-o(1)$. However  by Lemma~\ref{lem:prod-linear}, $\mrank(\pcm_{p_{lin}^\varphi})=\rank(M_{p_{lin}^\varphi}) = 2^{\Omega(m)}$ with probability $1-2^{-\Omega(m)}$, a contradiction. Hence $s\ge 2^{N^{1/4}}$.
\end{proof}

\paragraph*{Proof of Theorem~\ref{thm:sum-product-rop-perm}}
\begin{proof}
Suppose $s < 2^{N^{1/4}}$. Then by Lemma~\ref{lem:rank-sum-product}, Probability that there is an $f_{i,j}$ with $\rank(M_{f_{i,j}^\varphi}) \ge 2^{N^{4/15}}$ is at most $2^{-\Omega(m)} s\cdot t =o(1)$. By Lemma~\ref{sub-aditivity2} and \ref{sub-multiplicativity2}, we have  $\mrank(\pcm_{\perm_{n}^\varphi}) \le 2^{N^{4/15}}$  with probability $1-o(1)$. However, by Theorem~\ref{thm:perm-rank},  $\mrank(\pcm_{\perm_{n}^\varphi})=\rank(\perm_{n}^\varphi) = 2^{\Omega(m)}$ with probability $(1-1/n^2)/2$, a contradiction. Hence $s\ge 2^{N^{1/4}}$.
\end{proof}

\section*{Acknowledgements} We thank anonymous reviewers of an earlier version of the paper  for suggestions which improved the presentation.   Further, we thank one of the  anonymous reviewers  for pointing out an outline of argument for Lemma~\ref{lem:rank-upper bound}. 

\bibliographystyle{alpha}
\bibliography{refbib}

\end{document}